\documentclass[hidelinks,onefignum,onetabnum,nohypdvips]{siamart251216}

\usepackage[T1]{fontenc}
\usepackage{amsmath,amssymb,amsfonts}
\usepackage{bm}
\usepackage{graphicx}
\usepackage{booktabs}
\usepackage{tabularx}
\usepackage{threeparttable}
\usepackage{dcolumn}
\usepackage{multirow}
\usepackage{array}
\usepackage{makecell}
\usepackage{float}
\usepackage{algorithm}
\usepackage{algpseudocode}
\usepackage{chemformula}
\usepackage{xcolor}
\usepackage{enumerate}
\usepackage{lipsum}

\ifpdf
  \DeclareGraphicsExtensions{.eps,.pdf,.png,.jpg}
\else
  \DeclareGraphicsExtensions{.eps}
\fi

\newcommand{\bbR}{\mathbb{R}}

\newcommand{\dd}{\mathrm{d}}

\allowdisplaybreaks[4]
\newsiamthm{assumption}{Assumption}
\newsiamremark{remark}{Remark}
\newsiamremark{hypothesis}{Hypothesis}
\crefname{hypothesis}{Hypothesis}{Hypotheses}
\newsiamthm{claim}{Claim}
\newsiamremark{fact}{Fact}
\crefname{fact}{Fact}{Facts}

\headers{Generalized multi-species collision operators}{Y. Zhao, A. J. Christlieb and H. Lei}

\title{Learning a general class of admissible multi-species collision operators from molecular dynamics
\thanks{Submitted to the editors DATE.
\funding{This work was funded in part by the National Science Foundation under Grant DMS-2143739, the Department of Energy under Grant DOE-DESC0023164 and the ACCESS program through allocation MTH210005.}
}
}

\author{
Yue Zhao\thanks{
Department of Computational Mathematics, Science \& Engineering,
Michigan State University, 428 S Shaw Ln, East Lansing, MI 48824, USA
(\email{zhaoyu14@msu.edu}).
}
\and
Andrew J. Christlieb \thanks{
Department of Computational Mathematics, Science \& Engineering,
Michigan State University, 428 S Shaw Ln, East Lansing, MI 48824, USA;
Department of Mathematics,
Michigan State University, 619 Red Cedar Road, 
East Lansing, MI 48824, USA
(\email{christli@msu.edu}).
}
\and
Huan Lei\thanks{
Department of Computational Mathematics, Science \& Engineering,
Michigan State University, 428 S Shaw Ln, East Lansing, MI 48824, USA;
Department of Statistics \& Probability,
Michigan State University, 619 Red Cedar Road, 
East Lansing, MI 48824, USA
(\email{leihuan@msu.edu}).
}
}

\ifpdf
\hypersetup{
  pdftitle={
  Learning a general class of admissible multi-species collision operators from molecular dynamics
  },
  pdfauthor={Y. Zhao, A. Christlieb and H. Lei}
}
\fi

\begin{document}

\maketitle

\begin{abstract}
We develop a structure-preserving, data-driven collision operator for spatially homogeneous multi-species kinetic systems from molecular dynamics (MD).
The operator consists of diagonal self-collision blocks and ordered off-diagonal cross-species blocks to describe intra- and inter-species momentum and energy exchange. 
Within a local and point-wise identifiable kernel class, we develop the necessary and sufficient condition for the admissible kernel class satisfying the conservation laws, the H-theorem, and the frame indifference. 
Unlike the classical Landau operator, the off-diagonal kernels are not restricted to be symmetric under permutation of the two velocity variables. 
This unique structural freedom captures the distinct responses of different species to unresolved correlations and many-body effects arising from micro-scale particle interactions.
The equivalent parametrizable kernel formalization enables us to learn a generalized data-driven collision operator directly from MD, where the low-rank tensor representations and random sampling are used to achieve efficient kernel training and numerical simulation.
Numerical experiments show that the learned operator accurately predicts transport coefficients and the non-equilibrium relaxation, while retaining discrete conservation and entropy production. In particular, 
it captures plasma kinetics in the moderately coupled regime, where the predictions of both the Landau and the data-driven model restricted to velocity-permutation symmetry show significant discrepancies.
\end{abstract}

\begin{keywords}
Kinetic equation, collision operator, multi-species, data-driven modeling, structure-preserving, molecular dynamics
\end{keywords}

\begin{MSCcodes}
82C40, 82D10, 65M70, 82M37
\end{MSCcodes}

\section{Introduction}

The collision operator provides a canonical form for the kinetic description of the irreversible dissipation of plasmas, rarefied gases, and other many-particle systems. 
By encoding the effects of unresolved momentum and energy exchange among particles, the collision operator determines the relaxation and equilibration of the one-particle distribution function under conservation laws and physical constraints. 
Several classical operator forms, including the Landau \cite{landau1937kinetic}, the Boltzmann \cite{boltzmann1872weitere, wild1951boltzmann}, the Balescu-Lenard operator \cite{lenard1960bogoliubov, balescu1960irreversible}, as well as the Rosenbluth-Fokker-Planck formulation \cite{rosenbluth1957fokker}, 
have been proposed based on various approximations of the micro-scale particle correlations and interactions.

For plasmas with long-range Coulomb interactions, the Landau operator has been broadly used for systems in the weakly coupled regime. 
It can be derived from the Boltzmann operator in the grazing-collision limit \cite{alexandre2004landau} or from the Balescu-Lenard operator under suitable weak-screening assumptions. 
Extensive studies have been devoted to developing efficient and structure-preserving numerical methods for the one-species Landau equation, such as entropy-based discretizations \cite{degond1994entropy,buet1998conservative,buet1999numerical}, fast solvers based on Fourier spectral methods \cite{pareschi2000fast,filbet2002numerical,zhang2017conservative}, Hermite expansions \cite{li2020approximation,li2021hermite}, multigrid and multipole methods \cite{buet1997fast,lemou1998multipole}, and Rosenbluth-potential solvers \cite{chacon2000implicit}. 
Complementary particle methods include Monte Carlo binary collision \cite{takizuka1977binary,nanbu1997theory, nanbu1998weighted, caflisch2008hybrid} and regularized variational form \cite{carrillo2020particle,carrillo2021random,bailo2024collisional} have been proposed.

Beyond the one-species setting, multi-species kinetic models are essential for modeling momentum and energy exchange among distinct particle populations in systems such as dilute gas mixtures \cite{andries2002consistent}, strongly coupled ionic mixtures \cite{Sprenkle_Murillo_Nature_Comm_2022}, and fuel plasmas in inertial confinement fusion \cite{AmendtEtAl2010}. 
The multi-species Landau operator models the inter-species interactions in a pairwise form where the off-diagonal blocks share the same prescribed kernel formula as the one-species operator, using the corresponding pairwise coefficients. 
Efficient energy-conserving schemes \cite{taitano2015mass,taitano2016adaptive,hager2016fully}
and particle methods \cite{carrillo2024particle} have been developed. In addition, simplified collision models based on the BGK \cite{bhatnagar1954model}, Lenard–Bernstein \cite{lenard1958plasma}, and Dougherty \cite{dougherty1964model} formulations replace the full nonlocal collision operator with relaxation or drift–diffusion closures. Their multi-species extensions \cite{andries2002consistent,haack2017conservative,dougherty1967model,habbershaw2025nonlinear} incorporate cross-species momentum and energy exchange while enforcing conservation laws and entropy production.

Despite their broad use, classical multi-species collision models remain largely based on weak-coupling assumptions, including independent binary interactions, negligible particle correlations, and isotropic energy transfer. 
These assumptions are restrictive in realistic applications, such as inertial confinement fusion \cite{Rinderknecht_kinetic_plasma_reivew_2018, Indirect_drive_ICF_PRL_1_2024}, dense \cite{Dharma_Perrot_PRE_1998} and ultra-cold plasma \cite{Killian_Science_2007, Bergeson_Murillo_PoP_2025}, where the screening, local structure, and many-body correlations can substantially affect collisional relaxation and transport. 
Recent data-driven approaches \cite{zhao2025data, zhao2026fast, zhao2026molecular} learn a generalized collision operator for one-species plasma systems directly from micro-scale molecular dynamics (MD) simulations for one-species plasma across weakly and moderately coupled regimes. 
The constructed collision kernel takes an anisotropic and non-stationary form and captures the heterogeneous interactions arising from the micro-scale particle correlations and many-body effects. However, for multi-species plasmas, constructing such an operator raises a structural question absent from the one-species setting. Within a pairwise quadratic collision framework, the central issue is to identify the general admissible class of cross-species kernels consistent with reciprocal inter-species exchange, the conservation laws, the H-theorem, and frame indifference. Characterizing this admissible kernel class remains largely open.

In this work, we address this structural problem by characterizing a general admissible class of collision kernels for multi-species kinetic systems within a pairwise quadratic collision framework and directly learning kernels in this class from microscale MD simulations.
The operator is represented by a quadratic bracket within a metriplectic formulation \cite{morrison1984bracket, morrison1986paradigm}, where the collision kernel takes a pairwise block structure to model both the intra- and inter-species interactions. We then characterize the general admissible kernel class by establishing necessary and sufficient conditions
for the conservation of the species-mass, total-momentum, and kinetic energy, the H-theorem, the common-Maxwellian stationarity, and the frame-indifference. 
Crucially, the off-diagonal kernel takes a more general form than the diagonal kernel in the sense that the off-diagonal kernel does not need to satisfy the velocity permutation (i.e., exchange) symmetry between the two colliding particles. 
This general structural freedom can neither be captured by the classical Landau \cite{landau1937kinetic} nor by the one-species data-driven operator \cite{zhao2025data,zhao2026fast}. 
It is a unique feature of the present model, which allows the cross-species kernels to represent distinct responses of different species to unresolved particle correlations and many-body effects.
With the admissible kernel class, we derive an equivalent parameterizable representation and further introduce low-rank separated representations, which restore the convolution structure and can be efficiently learned from MD simulations.

The constructed data-driven collision operator, named DDCO, strictly preserves the physical constraints at the continuum level and therefore enables structure-preserving discretizations of the kinetic equation. We validate the learned operator in both near-equilibrium and nonequilibrium regimes. The predicted transport coefficients agree well with MD estimates, which verifies the accuracy of the present model for the linear response in the near-equilibrium regime.  The model also accurately predicts the non-equilibrium kinetic relaxation of the individual species in comparison with the MD results.  To examine the role of the generalized cross-species structure identified above, we further compare with a data-driven operator whose off-diagonal kernels are constrained to the velocity-permutation symmetry. The resulting predictions show significant discrepancies, which validates the present generalized inter-species collision form with the structural freedom of permutation symmetry.

The rest of this paper is organized as follows.
\Cref{sec:method} presents the generalized multi-species collision operator, including the necessary and sufficient conditions for the admissible kernel class and the equivalent parameterizable form, the low-rank kernel representation, the MD-based learning procedure, the structure-preserving discretization, and the transport formulation. 
\Cref{sec:numerical} reports the numerical results of transport coefficients and the non-equilibrium kinetic relaxation predicted from the present data-driven multi-species collision operator and the full MD results. Summaries and discussion are presented in \Cref{sec:summary}.

\section{Methods}\label{sec:method}

\subsection{Physical and normalized 0D-3V formulations}

Consider a spatially homogeneous plasma mixture of $S$ species.
Species $s$ has velocity distribution $f_{s}(\bm{v}, t)$, number density $n_{s}$, mass $m_{s}$, and charge $q_{s}$.
We separate the physical number density from the normalized velocity probability density function (PDF) $\eta_{s}$:
\begin{equation}\label{eq:physical-normalization}
	f_{s} = n_{s} \eta_{s}, \quad \int_{\bbR^{3}} \eta_{s}(\bm{v}, t) \dd \bm{v} = 1.
\end{equation}
Throughout the paper, a prime denotes evaluation at the second velocity: $f_{t}' = f_{t}(\bm{v}', t)$ and $\nabla' = \nabla_{\bm{v}'}$.
The physical collision equation is
\begin{subequations}\label{eq:dimensional-f-system}
\begin{align}
	\partial_{t} f_{s} & = \sum_{t = 1}^{S} Q_{st}[f_{s}, f_{t}], \label{eq:dimensional-f-equation}\\
	Q_{st}[f_{s}, f_{t}] & = \nabla_{\bm{v}} \cdot \int_{\bbR^{3}} \dfrac{\bm{\omega}_{st}}{m_{s}} \left[ \dfrac{f_{t}'}{m_{s}} \nabla_{\bm{v}} f_{s} - \dfrac{f_{s}}{m_{t}} \nabla_{\bm{v}'} f_{t}' \right] \dd \bm{v}'. \label{eq:dimensional-f-block}
\end{align}
\end{subequations}

\paragraph{Nondimensionalization}
We introduce the dimensionless variables 
\begin{equation}\label{eq:nondimensional-variables}
\begin{aligned}
	\widetilde{\bm{v}}& = \bm{v}/V_{0}, & \widetilde t& = t/t_{0}, &  \\
	\widetilde \eta_{s}(\widetilde{\bm{v}}, \widetilde t) & = V_{0}^{3} \eta_{s}(\bm{v}, t), & \widetilde{\bm{\omega}}_{st}& = \bm{\omega}_{st}/\omega_{0},
\end{aligned}
\end{equation}
where  $m_0$,  $t_0$, $V_0$, and $\omega_{0} = m_{0}^{2} V_{0}^{2}/n_{0} t_{0}$ are the characteristic scales of the mass, time, velocity, and the collision kernel.
In the remaining of the paper, we suppress all tildes and the dimensionless equation takes the form
\begin{subequations}\label{eq:multi-kinetic}
\begin{align}
	\partial_{t} \eta_{s} & = \sum_{t = 1}^{S} n_{t} \mathcal C_{st}[\eta_{s}, \eta_{t}], \label{eq:eta-kinetic-equation}\\
	\mathcal C_{st}[\eta_{s}, \eta_{t}] & = \nabla_{\bm{v}} \cdot \int_{\bbR^{3}} \dfrac{\bm{\omega}_{st}(\bm{v}, \bm{v}')}{m_{s}} \left[ \dfrac{\eta_{t}'}{m_{s}} \nabla_{\bm{v}} \eta_{s} - \dfrac{\eta_{s}}{m_{t}} \nabla_{\bm{v}'} \eta_{t}' \right] \dd \bm{v}', \label{eq:collision-block}
\end{align}
\end{subequations}
where $n_t$ and $m_s$ are constant and the kernel $\bm{\omega}_{st}(\bm{v}, \bm{v}')$ depends on the physical conditions of the multi-species plasma system. Rather than taking empirical forms such as the Landau model \cite{landau1937kinetic}, we aim to learn a generalized kernel $\bm{\omega}_{st}(\bm{v}, \bm{v}')$ directly from micro-scale molecular dynamics. 

\subsection{Multi-species generalized collision kernel}

To impose Galilean invariance in the kinetic model \cref{eq:multi-kinetic}, all collision blocks are evaluated using the same barycentric velocities.
Specifically, the species and barycentric velocities are
\begin{equation}\label{eq:common-mixture-velocity}
	\overline{\bm{u}} = \dfrac{\sum_{s = 1}^{S} m_{s} n_{s} \bm{u}_{s}}{\sum_{s = 1}^{S} m_{s} n_{s}}, \quad \bm{u}_{s} = \int_{\bbR^{3}} \bm{v} \eta_{s}(\bm{v}) \dd \bm{v} .
\end{equation}
Individual block $\bm\omega_{st}(\boldsymbol v, \boldsymbol v')$ is evaluated with the same peculiar velocities
\begin{equation}\label{eq:common-centered-velocities}
	\bm{c} = \bm{v} - \overline{\bm{u}}, \quad \bm{c}' = \bm{v}' - \overline{\bm{u}},
\end{equation}
where $\bm{c}$ and $\bm{c}'$ are invariant since $\bm v$ and $\overline{\bm u}$ acquire the same translation. 
Therefore, we further rewrite $\bm\omega_{st}(\boldsymbol v, \boldsymbol v')$ 
as 
$\bm\omega_{st}(\boldsymbol c, \boldsymbol c')$.

\subsubsection{Admissible kernel class}
To construct $\bm\omega_{st}(\boldsymbol c, \boldsymbol c')$, model \cref{eq:multi-kinetic} needs to strictly preserve physical constraints such as the conservation laws and the H-theorem, as well as the frame-indifferent symmetries. This motivates us to identify the admissible kernel class satisfying these constraints. 
In particular, we assume that the species distributions $\eta_s$ are sufficiently regular and decaying in velocity space, and that the admissible kernel representation is locally identifiable in the sense that localized weak identities holding for all admissible species distributions determine the corresponding kernel blocks almost everywhere.

\begin{proposition}[Necessary and sufficient kernel conditions]\label{prop:continuous-structure}
Within the admissible kernel class above, 
the following statements are equivalent.

\noindent
\textnormal{(i)}
For every admissible state, the collision operator preserves the conservation of species mass, total momentum and total kinetic energy
\begin{equation}
	\mathcal{M}_{s} = m_{s} n_{s} \int \eta_{s} \dd \bm{v} , \quad
    \bm{P} = \sum_{s} m_{s} n_{s} \int \bm{v} \eta_{s} \dd \bm{v} , \quad
    \mathcal{E} = \sum_{s} m_{s} n_{s} \int \dfrac{|\bm{v}|^{2}}{2} \eta_{s} \dd \bm{v} ,
\label{eq:mass_momentum_energy}
\end{equation}
non-negative entropy production 
\begin{equation}\label{eq:physical-entropy}
	\mathcal{S}[\bm{\eta}] = - \sum_{s=1}^{S}n_{s} \int_{\bbR^{3}} \eta_{s} \log \eta_{s} \dd \bm{v}, \qquad \dot{\mathcal{S}} \ge 0,
\end{equation}
the normalized Maxwellians
\begin{equation}\label{eq:common-maxwellian}
	\eta^M_{s}(\bm{v}) = \left(\dfrac{m_{s}}{2\pi k_{B}T}\right)^{3/2} \exp\left( -\dfrac{m_{s}|\bm{v}-\overline{\bm{u}}|^{2}}{2k_{B}T} \right)
\end{equation}
as stationary state, and frame-indifferent evolution in $O(3)$.

\noindent
\textnormal{(ii)}
For almost every $(\bm{c},\bm{c}')$, the kernel satisfies
\begin{subequations}\label{eq:kernel-conditions}
\begin{align}
	\bm{\omega}_{st}(\bm{c}, \bm{c}') & = \bm{\omega}_{st}(\bm{c}, \bm{c}')^{T} \succeq 0, \label{eq:pos-semidef}\\
	\bm{\omega}_{st}(\bm{c}, \bm{c}')(\bm{c} - \bm{c}')& = \bm{0}, \label{eq:orthogonality}\\
	\bm{\omega}_{st}(\bm{c}, \bm{c}') & = \bm{\omega}_{ts}(\bm{c}', \bm{c}), \label{eq:reciprocity}\\
    \bm{\omega}_{st}(\bm{U} \bm{c}, \bm{U} \bm{c}') & = \bm{U} \bm{\omega}_{st}(\bm{c}, \bm{c}') \bm{U}^{T}, \quad \bm{U}\in O(3),
    \label{eq:rotational-symmetry}
\end{align}
\end{subequations}

\end{proposition}

\begin{proof} Let us assume condition \textnormal{(ii)} holds. For arbitrary smooth test functions $\varphi_{s}$, integration by parts and the reciprocal pairing $(s,t,\bm{v},\bm{v}')\leftrightarrow(t,s,\bm{v}',\bm{v})$ give
\begin{equation}
\label{eq:continuous-weak-identity}
\begin{split}
&\sum_{s=1}^{S}n_{s}\int\varphi_{s}\partial_{t}\eta_{s}\dd\bm{v} 
= - \sum_{s,t=1}^{S}n_{s}n_{t} \iint \left(\dfrac{\nabla_{\bm{v}} \varphi_s}{m_{s}}\right)^{T} \bm{\omega}_{st}(\bm{c}, \bm{c}') \left[ \dfrac{\eta_{t}'}{m_{s}} \nabla_{\bm{v}} \eta_{s} - \dfrac{\eta_{s}}{m_{t}} \nabla_{\bm{v}'} \eta_{t}' \right]
\dd\bm{v}\dd\bm{v}' \\
&= - \frac{1}{2} \sum_{s,t=1}^{S}n_{s}n_{t} \iint \left[ \dfrac{\nabla_{\bm{v}} \varphi_s}{m_{s}} - \dfrac{\nabla_{\bm{v}'} {\varphi_t}'}{m_{t}} \right]^{T}  \bm{\omega}_{st}(\bm{c}, \bm{c}') \left[ \dfrac{\eta_{t}'}{m_{s}} \nabla_{\bm{v}} \eta_{s} - \dfrac{\eta_{s}}{m_{t}} \nabla_{\bm{v}'} \eta_{t}' \right]
\dd\bm{v}\dd\bm{v}' \\
&=- \frac{1}{2} \sum_{s,t=1}^{S}n_{s}n_{t} \iint \eta_{s}\eta_{t}' \bm{D}_{st}[\varphi]^{T} \bm{\omega}_{st} \bm{D}_{st}[\log\eta] \dd\bm{v}\dd\bm{v}',
\end{split}
\end{equation}
where $\bm{D}_{st}[h] = \nabla h_{s}/m_{s} - \nabla' h_{t}'/m_{t}$ and the second identity use condition \cref{eq:reciprocity}.

In particular, by choosing  $\varphi_{s}=m_{s} \delta_{sr}$, $m_{s}\bm v$, and $m_{s}|\bm{v}|^{2}/2$, $\bm{D}_{st}[\varphi_{s}]$ yields $\bm{0}$, $\bm{0}$, and $\bm{c}-\bm{c}'$, respectively, which prove the conservation of species $r$ mass, total momentum, and total kinetic energy defined in Eq. \cref{eq:mass_momentum_energy}. 
Furthermore, by taking $\varphi_{s} = -(1+\log\eta_{s})$, we have Eq. \cref{eq:physical-entropy}. 
Finally, using $\bm{D}_{st}[\log \eta^M] = -\frac{\bm{c}-\bm{c}'}{k_{B}T}$ and the null condition \cref{eq:orthogonality}, we can show Eq. \cref{eq:common-maxwellian} is a stationary state.

For an orthogonal matrix $\bm{U}$, define $\eta_{s}^{\bm{U}}(\bm{v}) = \eta_{s}(\bm{U}^{T}\bm{v})$.
\Cref{eq:rotational-symmetry} and an orthogonal change of variables give
\begin{equation}\label{eq:operator-frame-indifference}
    \mathcal{C}_{st}[\eta_{s}^{\bm{U}}, \eta_{t}^{\bm{U}}](\bm{v}) = \mathcal{C}_{st}[\eta_{s}, \eta_{t}](\bm{U}^{T} \bm{v}) ,
\end{equation}
which proves the frame-indifference of model \cref{eq:multi-kinetic} in $O(3)$. 

Conversely, we assume condition \textnormal{(i)} holds. The pairwise momentum evolution gives
\[
\dot{\bm P}_{st} + \dot{\bm P}_{ts} = n_s n_t \iint \left[ \bm{\omega}_{ts}(\bm c',\bm c) - \bm{\omega}_{st}(\bm c,\bm c') \right] \bm{D}_{st}[\log \eta]\, \eta_s \eta_t' \, \dd \bm{v} \, \dd \bm{v}'.
\]
The conservation of the total momentum for arbitrary localized \(\eta_s,\eta_t\) requires the conservation of the pairwise momentum and therefore condition \cref{eq:reciprocity}. 

Similarly, we consider the pairwise energy exchange
\[
\dot{\mathcal{E}}_{st} + \dot{\mathcal{E}}_{ts} = - n_s n_t \iint (\boldsymbol{c} - \boldsymbol{c}')^T
\bm{\omega}_{st}(\bm{c},\bm{c}') \bm{D}_{st}[\log \eta]\, \eta_s \eta_t' \, \dd \bm{v} \, \dd \bm{v}'.
\]
The total energy conservation for 
arbitrary localized \(\eta_s,\eta_t\) requires $\dot{\mathcal{E}}_{st}
+
\dot{\mathcal{E}}_{ts} = 0$, which yields the left-null condition
\begin{equation}\label{eq:left-null-condition}
	\bm{\omega}_{st}(\bm{c},\bm{c}')^{T} (\bm{c}-\bm{c}') = \bm{0}.
\end{equation}
By Eq. \cref{eq:continuous-weak-identity}, the common-Maxwellian stationarity gives the right-null condition \cref{eq:orthogonality} and 
entropy monotonicity gives
\begin{equation}\label{eq:symmetric-part-positive}
	\operatorname{Sym}\bm{\omega}_{st} = \frac{\bm{\omega}_{st}+\bm{\omega}_{st}^{T}}{2} \succeq0.
\end{equation}
Furthermore, by $\nabla_{\bm v} \eta_{s}^{\bm{U}}(\bm{v}) = \bm U \nabla \eta_s(\bm{U}^{T} \bm v)$, 
frame indifference of model \cref{eq:multi-kinetic} gives condition \cref{eq:rotational-symmetry}.

Finally, we show that $\bm{\omega}_{st} = \bm{\omega}_{st}^T$.  Write $\bm{\omega}_{st}=\bm{S}_{st}+\bm{A}_{st}$ with $\bm{S}_{st}^{T}=\bm{S}_{st}$ and $\bm{A}_{st}^{T}=-\bm{A}_{st}$.
The two null conditions imply $\bm{A}_{st}(\bm{c}-\bm{c}')=\bm{0}$.
In three dimensions, there is a unique axial vector $\bm{a}_{st}$ such that $\bm{A}_{st}\bm{x}=\bm{a}_{st}\times\bm{x}$.
Under $\bm{U}\in O(3)$ it transforms as $\bm{a}_{st}(\bm{U}\bm{c},\bm{U}\bm{c}') = \det(\bm{U}) \bm{U} \bm{a}_{st} (\bm{c},\bm{c}')$.
Reflection across $\operatorname{span}\{\bm{c},\bm{c}'\}$ leaves $\bm{c}$ and $\bm{c}'$ unchanged and gives $\bm{a}_{st}=\alpha_{st}\bm{c}\times\bm{c}'$.
In contrast, the null condition also makes $\bm{a}_{st}$ parallel to $\bm{c}-\bm{c}'$, so $\bm{a}_{st}=\bm{0}$ almost everywhere.
Thus $\bm{\omega}_{st}$ is symmetric and \cref{eq:symmetric-part-positive} reduces to \cref{eq:pos-semidef} and Eq. \cref{eq:left-null-condition} reduces to \cref{eq:orthogonality}.
\end{proof}

\begin{remark}
Furthermore, suppose that the species-pair blocks are connected and 
\begin{equation}\label{eq:strict-tangential-coercivity}
	\operatorname{Null} \left( \bm{\omega}_{st}(\bm{c},\bm{c}') \right) = \operatorname{span}\{\bm{c}-\bm{c}'\}
\end{equation}
for almost every $\bm{c}\neq\bm{c}'$. 
It is straightforward to show model \cref{eq:multi-kinetic} achieves a unique equilibrium state if and only if all species have the common-Maxwellian form \cref{eq:common-maxwellian}, determined by the conserved species densities, total momentum, and total energy.
\end{remark}

Condition \cref{eq:kernel-conditions} identifies a general class of the collision kernel. Compared with the classical Landau model $\bm\omega^L_{st}(\bm c, \bm c') \propto \mathcal{P} \vert \bm c-\bm c'\vert^{-1}$ and $\mathcal{P} = \bm I - \hat{\bm u}\hat{\bm u}^T$, there are two major differences. First, the collision kernel $\bm\omega_{st}(\bm c, \bm c')$ can be generally non-stationary and anisotropic. Second, the inter-species collision kernel does not need to satisfy the permutation symmetry, i.e., for off-diagonal blocks, the admissible class permits 
\begin{equation}
\bm\omega_{st}(\bm c, \bm c') \neq \bm\omega_{st}(\bm c', \bm c), \qquad s\neq t.
\label{eq:permutation_non_symmetry}
\end{equation}
This condition arises from the microscale particle correlations and many-body effects. The collision between the two particles will be further affected by the unresolved environment, where the responses can be different between $(s, t, \bm c, \bm c')$ and $(s, t, \bm c', \bm c)$. 
As shown in \Cref{sec:numerical}, the present generalized collision operator enables us to accurately predict the multi-species plasma kinetics in moderately coupled regimes with non-negligible particle correlations, where the Landau model shows limitations.

\subsubsection{Equivalent parameterizable representation} We next parameterize the admissible kernels satisfying \Cref{prop:continuous-structure} and define
\begin{equation}
	\mathcal{P} = \bm{I} - \bm{u} \bm{u}^{T}/{|\bm{u}|^{2}}, \quad 
    \bm{u} = \bm{c} - \bm{c}' , \quad
    \bm{r} = \bm{c} + \bm{c}' .
\end{equation}

\begin{proposition}[Equivalent kernel formalization]\label{prop:formalization}
Let $\bm{\omega}_{st}$ depend only on $\bm{c}$, $\bm{c}'$, and scalar state variables. 
For almost every pair with $\bm{u}\neq\bm{0}$ and $\mathcal{P} \bm{r} \neq \bm{0}$, the following statements are equivalent.

\noindent
\textnormal{(i)}
The kernel satisfies all conditions in \cref{eq:kernel-conditions}.

\noindent
\textnormal{(ii)}
There exist real scalar encoders $g_{st,1}$ and $g_{st,2}$ such that
\begin{subequations}\label{eq:two-encoder-kernel}
\begin{align}
    \bm{\omega}_{st} &= g_{st,1}^{2}|\mathcal{P}\bm{r}|^{2}\mathcal{P} + \left(g_{st,2}^{2}-g_{st,1}^{2}\right) \mathcal{P}\bm{r}(\mathcal{P}\bm{r})^{T}, \label{eq:encoder-representation} \\
	g_{ts,\ell}(\bm{c}',\bm{c})^{2} &= g_{st,\ell}(\bm{c},\bm{c}')^{2}, \label{eq:encoder-reciprocity} \\
	g_{st,\ell}(\bm{U}\bm{c},\bm{U}\bm{c}')^{2} &= g_{st,\ell}(\bm{c},\bm{c}')^{2}, \quad
	\bm{U}\in O(3), \quad \ell=1,2. \label{eq:encoder-covariance}
\end{align}
\end{subequations}
\end{proposition}

\begin{proof}

Let us assume \textnormal{(i)} holds. The symmetry \cref{eq:pos-semidef} and the null condition \cref{eq:orthogonality} imply $\bm \omega_{st}= \mathcal{P} \bm \omega_{st} \mathcal{P} \succeq 0$. Define $\hat{\bm r} = \mathcal{P} \bm r/ \vert \mathcal{P} \bm r \vert$ and $\hat{\bm h} = \bm \hat{\bm u} \times \hat{\bm r}$, $\bm \omega_{st}$ takes the general form 
\[
\boldsymbol{\omega}(\boldsymbol{c}, \boldsymbol{c}') = \mathcal{P}\left (\lambda_{\parallel} \hat{\boldsymbol r}\hat{\boldsymbol r}^T + \lambda_{c} \left(\hat{\boldsymbol r} \hat{\boldsymbol h}^T +  \hat{\boldsymbol h} \hat{\boldsymbol r}^T\right) + \lambda_{\perp}  \hat{\boldsymbol h}\hat{\boldsymbol h}^T \right) \mathcal{P}, \qquad \begin{pmatrix}
\lambda_{\parallel} &\lambda_{c}\\
\lambda_{c}&\lambda_{\perp}
\end{pmatrix} \succeq 0
\]
Using reflection symmetry $\boldsymbol{\omega}(\boldsymbol{c}, \boldsymbol{c}') = \boldsymbol{\omega}(-\boldsymbol{c}, -\boldsymbol{c}')$, it is easy to show $\lambda_{c} \equiv 0$. 
Therefore, we can re-write $\bm \omega_{st}$ as 
\[
\bm \omega_{st} = \mathcal{P} ( \lambda_{\parallel} \hat{\bm r}\hat{\bm r}^T  + \lambda_{\perp} (\bm I -  \hat{\bm u}\hat{\bm u}^T - \hat{\bm r}\hat{\bm r}^T))\mathcal{P} = \lambda_{\perp}  \mathcal{P} + (\lambda_{\parallel} - \lambda_{\perp}) \hat{\bm r}\hat{\bm r}^T, \quad \lambda_{\parallel}, \lambda_{\perp} \ge 0. 
\]
Taking
\begin{equation}\label{eq:encoder-eigenvalues}
	g_{st,1}^{2} = \lambda_{\perp}/|\mathcal{P}\bm{r}|^{2}, \quad 
    g_{st,2}^{2} = \lambda_{\parallel}/|\mathcal{P}\bm{r}|^{2}
\end{equation}
gives \cref{eq:encoder-representation}.
Kernel reciprocity \cref{eq:reciprocity} and rotational symmetry \cref{eq:rotational-symmetry} require these eigenvalues satisfy \cref{eq:encoder-reciprocity,eq:encoder-covariance}.

Conversely, let us assume \textnormal{(ii)} holds. The eigenvalues condition \cref{eq:two-encoder-kernel} for $\bm \omega_{st}$ are $0$ along $\bm{u}$, $g_{st,2}^{2}|\mathcal{P} \bm{r}|^{2}$ along $\mathcal{P}\bm{r}$, and $g_{st,1}^{2}|\mathcal{P}\bm{r}|^{2}$ in the remaining tangential direction.
Thus the kernel $\bm \omega_{st}$ is symmetric positive semi-definite and perpendicular to $\bm{u}$.
\Cref{eq:encoder-reciprocity,eq:encoder-covariance} give ordered-pair reciprocity condition \cref{eq:reciprocity} and rotational symmetry \cref{eq:rotational-symmetry}, respectively.
\end{proof}

With the non-stationary form of $\bm\omega(\bm c, \bm c')$, direct evaluation of $C_{st}[\eta_s, \eta_t]$ in model \cref{eq:collision-block} loses the convolution structure and is computationally expensive. 
For efficient evaluation, we further represent each pair-ordered encoder function in low-rank spectral separation form.
The separation is imposed on scalar invariants and reciprocal blocks share parameters, so conditions \cref{eq:kernel-conditions} remain exact, i.e., 
\begin{equation}\label{eq:separation}
\begin{aligned}
    g_{ss,\ell}(\bm{c}, \bm{c}') =& \sum_{j=1}^{J_{ss,\ell}} L_{ss,\ell,j}(|\bm{c} - \bm{c}'|) \left[ M_{ss,\ell,j}(|\bm{c}|) N_{ss,\ell,j}(|\bm{c}'|) + N_{ss,\ell,j}(|\bm{c}|) M_{ss,\ell,j}(|\bm{c}'|) \right] , \\
    g_{st,\ell}(\bm{c}, \bm{c}') =& \sum_{j=1}^{J_{st,\ell}} L_{st,\ell,j}(|\bm{c} - \bm{c}'|) M_{st,\ell,j}(|\bm{c}|) N_{st,\ell,j}(|\bm{c}'|) , \quad s \neq t , \\
    g_{ts,\ell}(\bm c,\bm c') =& \sum_{j=1}^{J_{st,\ell}} L_{st,\ell,j}(|\bm c-\bm c'|) M_{st,\ell,j}(|\bm c'|) N_{st,\ell,j}(|\bm c|) , \quad g_{ts,\ell}(\bm c',\bm c) = g_{st,\ell}(\bm c,\bm c')
\end{aligned}
\end{equation}
where $\ell=1,2$, and diagonal blocks are exchange symmetric $g_{ss,\ell} (\bm c, \bm c') = g_{ss,\ell} (\bm c', \bm c)$, while off-diagonal blocks only need to satisfy the reciprocal condition but exchange symmetry does not hold in general, i.e., $g_{st,\ell} (\bm c, \bm c') = g_{ts,\ell} (\bm c', \bm c)$, $g_{st,\ell} (\bm c, \bm c') \neq g_{st,\ell} (\bm c', \bm c)$.

With the separation form \cref{eq:separation} of $\bm\omega_{st}$, block $\mathcal{C}_{st}$ in \cref{eq:collision-block} can be written as a finite sum of 3-dimensional velocity convolutions and can be evaluated by FFT without storing a dense 6-dimensional kernel \cite{zhao2026fast}.
For $N$ total velocity nodes, $S$ species, and a representative separation rank $J$, the resulting cost is $\mathcal{O}(S^{2} J^{2} N \log N)$.

\paragraph{Landau limit}
As a special case, by choosing $g_{st,1}^{2}=g_{st,2}^{2} \propto |\mathcal{P}\bm{r}|^{-2}|\bm{u}|^{-1}$, 
the generalized kernel reduces to the classical Landau form with stationary and exchange-symmetry structure, i.e., 
\begin{equation}
\begin{aligned}
    \bm{\omega}_{st}(\bm c, \bm c') &= \bm{\omega}_{st}(\bm c', \bm c) = \dfrac{\ln\Lambda_{st} q_{s}^{2} q_{t}^{2}} {8 \pi \epsilon_{0}^{2} |\bm{u}|} \mathcal{P}, \quad \Lambda_{st}=\lambda_D/b_{\min}^{st} , \\
    \lambda_{D} &= \left( \dfrac{\epsilon_{0} k_{B} T}{\sum_{s} n_{s} q_{s}^{2}} \right)^{1/2}, ~~\qquad \qquad b_{\min}^{st} = \dfrac{|q_{s} q_{t}|}{4 \pi \epsilon_{0} k_BT}.
\end{aligned}
\label{eq:landau_model}
\end{equation}
where $\ln\Lambda_{st}$ is the Coulomb logarithm.  In the weakly coupled limit, i.e., $\Lambda_{st} \gg 1$, the Landau form provides an accurate prediction and $\bm\omega_{st}$ remains stationary and isotropic on the plane orthogonal to $\bm u$. For a stronger coupled plasma system, $\Lambda_{st} \lesssim 1$ and the Landau form \cref{eq:landau_model} is generally insufficient to model the plasma kinetics. Instead, we aim to learn a generalized operator in the form of \cref{eq:multi-kinetic,eq:two-encoder-kernel,eq:separation} from micro-scale MD specified in the following \Cref{sec:weak_form_learning}.

\subsection{Weak-form learning}
\label{sec:weak_form_learning}

To learn the multi-species collision kernel $\bm\omega(\bm c, \bm c')$, we match the evolution of the particle velocity PDF $\eta_s(\bm v, t)$ predicted by the kinetic model \cref{eq:multi-kinetic} with the full MD simulations. 
In particular, directly matching the strong form requires the numerical evaluation of $\nabla_{\bm v}\eta_s(\bm v, t)$ from the measure-valued $\eta_s$ of the MD particles, which can be computationally intractable. 
Instead, we use a weak-form  \cite{zhao2025data,zhao2026fast} by projecting the prediction residual onto a set of smooth test functions $\psi_{k}(\bm{c})$, see in \Cref{app:md-learning}. 
The training loss is
\begin{equation}\label{eq:loss}
	\mathcal L = \sum_{j, k, s} \left| \left.\dfrac{\dd}{\dd t} \langle \eta_{s}, \psi_{k} \rangle \right|_{t=t_{j}, \rm MD} - \sum_{t = 1}^{S} n_{t} \left\langle \mathcal C_{st}[\eta_{s}^{j}, \eta_{t}^{j}], \psi_{k} \right\rangle \right|^{2} ,
\end{equation}
where $j$ denotes the snapshot time $t_{j}=j\Delta t$.
The barycentric velocity \cref{eq:common-mixture-velocity} is evaluated from the complete MD mixture and is used consistently in each block.

From $N_{s}$ MD particles in species $s$ at time $j$,
\begin{equation}
	\left.\dfrac{\dd}{\dd t} \langle \eta_{s}, \psi_{k} \rangle \right|_{\rm MD} \simeq \dfrac{1}{\Delta t N_{s}} \sum_{i = 1}^{N_{s}} \left[\psi_{k}(\bm{c}_{i}^{j + 1}) - \psi_{k}(\bm{c}_{i}^{j}) \right] ,
\end{equation}
where the empirical measure $\eta_{s}^{j}(\bm{c}) = N_{s}^{-1} \sum_{i=1}^{N_{s}} \delta(\bm{c} - \bm{c}_{i})$ from MD is used.

For a smooth test function $\psi_{k}$, the weak collision block is
\begin{equation}\label{eq:weak-collision-block}
\begin{aligned}
	\left\langle \mathcal C_{st}[\eta_{s}, \eta_{t}], \psi_{k} \right\rangle =& - \iint(\nabla \psi_{k})^{T} \dfrac{\bm{\omega}_{st}}{m_{s}} \left[ \dfrac{\eta_{t}'}{m_{s}} \nabla \eta_{s} - \dfrac{\eta_{s}}{m_{t}} \nabla' \eta_{t}' \right] \dd \bm{v}' \dd \bm{v} \\
    =& \dfrac{1}{m_{s}^{2}} \iint \eta_{s} \eta_{t}' \left[\bm{\omega}_{st}:\nabla \nabla \psi_{k} + (\nabla_{\bm{v}} \cdot \bm{\omega}_{st}) \cdot \nabla \psi_{k} \right] \dd \bm{v}' \dd \bm{v}\\
	& - \dfrac{1}{m_{s} m_{t}} \iint \eta_{s} \eta_{t}' (\nabla_{\bm{v}'}\!\cdot \bm{\omega}_{st}) \cdot \nabla \psi_{k} \dd \bm{v}' \dd \bm{v}.
\end{aligned}
\end{equation}

The integrals are expectations over independent samples $\bm{v} \sim \eta_{s}$ and $\bm{v}' \sim \eta_{t}$.
Since the MD trajectories contain $N_{s} \sim 10^{6}$ particles in each species, the expectations are estimated with random mini-batch sampling \cite{ketkar2017stochastic, li2014efficient, jin2020random, jin2021random}, i.e., 
\begin{equation}
\begin{aligned}
    \iint \eta_{s} \eta_{t}' \bm{\omega}_{st} : \nabla \nabla \psi_{k} \dd \bm{v}' \dd \bm{v} &= \dfrac{1}{N_{s} N_{t}} \sum_{i}^{N_{s}} \sum_{i'}^{N_{t}} \bm{\omega}(\bm{c}_{i}, \bm{c}_{i'} ') : \nabla \nabla \psi_{k}(\bm{c}_{i}), \\
    &\approx \dfrac{1}{P}\sum_{p=1}^{P} \bm{\omega}(\bm{c}_{i(p)}, \bm{c}_{i'(p)} ') : \nabla \nabla \psi_{k}(\bm{c}_{i(p)}),
\end{aligned}
\end{equation}
with pairs $(\bm{c}_{i(p)}, \bm{c}_{i'(p)} ')$ randomly selected from $(\eta_{s}, \eta_{t})$ in MD samples.

At each optimization iteration, one kernel block is differentiated per time while all remaining blocks are frozen. 
The gradients from the diagonal blocks and reciprocal cross-block losses are accumulated, followed
by a single optimizer update for all network parameters.
All encoder functions $L_{\ast,j}(|\bm{u}|)$, $M_{\ast,j}(|\bm{c}|)$, and $N_{\ast,j}(|\bm{c}'|)$: $\mathbb{R} \rightarrow \mathbb{R}$ are parameterized by fully connected neural networks with $5$ layers and $10$ neurons per hidden layer, where system state variables such as $n_{s}$, $T$, and $q_{s}$ are not used as input variables.
The networks are trained using Adam optimizer \cite{Kingma_Ba_Adam_2015} with initial learning rate $10^{-3}$, decay factor $0.99$ every $1000$ steps, batch size $P=10^5$, and $5 \times 10^5$ training steps.

\subsection{Numerical discretizations}
As shown in \Cref{prop:continuous-structure,prop:formalization}, the constructed collision operator $\mathcal{C}_{st}(\eta, \eta)$ strictly preserves the conservation laws and entropy production for the multi-species kinetic model \cref{eq:multi-kinetic} on the continuum level. This enables us to further develop structure-preserving numerical discretizations for model \cref{eq:multi-kinetic}. 
Let $\bm{v}_{\bm{k}} = (k_{1} h_{1}, k_{2} h_{2}, k_{3} h_{3})$ for $\bm{k} \in \mathbb{Z}^{3}$, and let $w_{h}=h_{1}h_{2}h_{3}$.
The centered gradient $\mathrm{G}$ and the divergence $\mathrm{D}$ are
\begin{equation}\label{eq:centered-differences}
	(\mathrm{G} z)_{\bm{k}, \alpha} = \frac{z_{\bm{k} + \bm{e}_{\alpha}} - z_{\bm{k} - \bm{e}_{\alpha}}}{2h_{\alpha}}, \quad
	(\mathrm{D}\bm{F})_{\bm{k}} = \sum_{\alpha=1}^{3} \frac{F_{\bm{k} + \bm{e}_{\alpha}, \alpha} - F_{\bm{k} - \bm{e}_{\alpha},\alpha}}{2h_{\alpha}}.
\end{equation}
For summable sequences with no boundary contribution, index shifts give
\begin{equation}\label{eq:centered-sbp}
	w_{h} \sum_{\bm{k}} z_{\bm{k}}(\mathrm{D}\bm{F})_{\bm{k}} = - w_{h} \sum_{\bm{k}}(\mathrm{G} z)_{\bm{k}}^{T}\bm{F}_{\bm{k}}.
\end{equation}
Define
\begin{equation}\label{eq:discrete-pair-difference}
	\bm{\Delta}[z]_{st, \bm{k}\bm{l}} = \frac{(\mathrm{G} z_{s})_{\bm{k}}}{m_{s}} - \frac{(\mathrm{G} z_{t})_{\bm{l}}}{m_{t}},
\end{equation}
and the centered-difference flux scheme
\begin{subequations}\label{eq:centered-flux-scheme}
\begin{align}
	\dot{\eta}_{s, \bm{k}} &= \sum_{t=1}^{S} n_{t}(\mathrm{D} \bm{p}_{st})_{\bm{k}} , \label{eq:centered-semidiscrete}\\
	\bm{p}_{st, \bm{k}} &= w_{h} \sum_{\bm{l}} \frac{\bm{\omega}_{st,\bm{k}\bm{l}}}{m_{s}} \eta_{s, \bm{k}}\eta_{t,\bm{l}} \bm{\Delta}[\log\eta]_{st,\bm{k}\bm{l}}. \label{eq:centered-flux}
\end{align}
\end{subequations}

\begin{proposition}[Forward-Euler conservation and semi-discrete entropy production] \label{prop:discrete-structure}
For the unbounded grid, we assume that the collision fluxes decay sufficiently fast, and $\eta_{s,\bm{k}}>0$.
The discrete collision kernel satisfies
\begin{equation}\label{eq:discrete-kernel-conditions}
	\bm{\omega}_{st, \bm{k}\bm{l}} = \bm{\omega}_{st, \bm{k}\bm{l}}^{T}\succeq 0, \quad
	\bm{\omega}_{st, \bm{k}\bm{l}} = \bm{\omega}_{ts, \bm{l}\bm{k}}, \quad
    \bm{\omega}_{st, \bm{k}\bm{l}} (\bm{v}_{\bm{k}} - \bm{v}_{\bm{l}}) = \bm{0}.
\end{equation}
The forward-Euler discretization update
\begin{equation}\label{eq:forward-euler}
	\eta_{s, \bm{k}}^{n+1} = \eta_{s, \bm{k}}^{n} + \Delta t \sum_{t=1}^{S} n_{t}(\mathrm{D} \bm{p}_{st}^{n})_{\bm{k}}
\end{equation}
exactly preserves each species mass, total momentum and total kinetic energy
\begin{equation}\label{eq:fully-discrete-invariants}
	\mathcal{M}_{s}^{n} = m_{s} n_{s} w_{h} \sum_{\bm{k}}\eta_{s,\bm{k}}^{n}, ~
	\bm{P}^{n} = w_{h} \sum_{s=1}^{S}m_{s}n_{s}\sum_{\bm{k}}\bm{v}_{\bm{k}}\eta_{s,\bm{k}}^{n}, ~
	\mathcal{E}^{n} = w_{h} \sum_{s=1}^{S}m_{s}n_{s}\sum_{\bm{k}}\frac{|\bm{v}_{\bm{k}}|^{2}}{2}\eta_{s,\bm{k}}^{n}.
\end{equation}
Moreover, the semi-discrete entropy satisfies
\begin{equation}\label{eq:discrete-entropy-law}
	\frac{\dd \mathcal{S}_{h}}{\dd t} = \frac{w_{h}^{2}}{2} \sum_{s, t=1}^{S} n_{s} n_{t} \sum_{\bm{k}, \bm{l}} \eta_{s,\bm{k}}\eta_{t,\bm{l}} \bm{\Delta}[\log\eta]_{st, \bm{k}\bm{l}}^{T} \bm{\omega}_{st, \bm{k}\bm{l}} \bm{\Delta}[\log\eta]_{st, \bm{k}\bm{l}} \geq 0,
\end{equation}
where
\begin{equation}\label{eq:discrete-entropy}
	\mathcal{S}_{h} = - w_{h} \sum_{s=1}^{S} n_{s} \sum_{\bm{k}} \eta_{s,\bm{k}} \log \eta_{s,\bm{k}}.
\end{equation}
\end{proposition}

\begin{proof}
\Cref{eq:centered-differences} satisfies \cref{eq:centered-sbp} and
\begin{equation}\label{eq:centered-polynomial-exactness}
	\mathrm{G} 1 = \bm{0}, \quad \mathrm{G} v_{\alpha} = \bm{e}_{\alpha}, \quad \mathrm{G} \frac{|\bm{v}|^{2}}{2} = \bm{v}.
\end{equation}
Pair $(s,t,\bm{k},\bm{l})$ with $(t,s,\bm{l},\bm{k})$, and apply \cref{eq:discrete-kernel-conditions} for every test function $\varphi_{s,\bm{k}}$ give
\begin{equation}\label{eq:discrete-weak-identity}
	w_{h} \sum_{s=1}^{S} n_{s} \sum_{\bm{k}} \varphi_{s,\bm{k}} \dot{\eta}_{s,\bm{k}} = 
    - \frac{w_{h} ^{2}}{2} \sum_{s,t=1}^{S}n_{s}n_{t} \sum_{\bm{k},\bm{l}} \eta_{s,\bm{k}}\eta_{t,\bm{l}} 
    \bm{\Delta}[\varphi]_{st,\bm{k}\bm{l}}^{T} \bm{\omega}_{st,\bm{k}\bm{l}} \bm{\Delta}[\log\eta]_{st,\bm{k}\bm{l}}.
\end{equation}
The test functions $\varphi_{s,\bm{k}} = m_{s} \delta_{sr}$, $m_{s}v_{\bm{k},\alpha}$, and $m_{s}|\bm{v}_{\bm{k}}|^{2}/2$ correspond to species-$r$ mass, total momentum and total kinetic energy, and yield $\bm{\Delta}[\varphi]_{st,\bm{k}\bm{l}} = \bm{0}$, $\bm{0}$, and $\bm{v}_{\bm{k}}-\bm{v}_{\bm{l}}$, respectively.
Therefore, the forward-Euler update satisfies the full discrete conservation laws, since the invariants are linear in the grid values.

Finally, taking $\varphi_{s,\bm{k}}=-(1+\log\eta_{s,\bm{k}})$ in \cref{eq:discrete-weak-identity} yields the semi-discrete entropy production \cref{eq:discrete-entropy-law}.
\end{proof}

\begin{remark} 
On a finite velocity grid, a sufficient boundary condition for the conservation and entropy arguments is the discrete no-flux condition
\begin{equation}
    \bm{\mathcal{F}}_{s} \cdot \bm{n}=0 \quad \text{on } \partial \Omega_{v}, \quad \bm{\mathcal{F}}_{s}=\sum_{t=1}^{S} n_{t} \bm{p}_{st},\quad s=1,\ldots,S,
\end{equation}
which makes the summation-by-parts boundary term vanish.
In addition, the boundary closure of $\mathrm{G}$ must reproduce the collision invariants \cref{eq:centered-polynomial-exactness}.
The discrete Maxwellian is stationary because $\bm\Delta[\log\eta^M]_{st,\bm k\bm l} = -\frac{\bm{v}_{\bm k}-\bm v_{\bm l}}{k_B T}$ lies in the null space of $\bm\omega_{st,\bm k\bm l}$.
The same conservation argument extends to Runge-Kutta, summation-by-parts, and finite-volume schemes by using conservative residuals and satisfying properties of the gradient-divergence pair.
\end{remark}

\subsection{Transport coefficients}
To further validate the accuracy of the multi-species collision operator, we compare the predicted transport coefficients in the hydrodynamic limit of the kinetic model \cref{eq:multi-kinetic} with full MD results. Specifically, we consider the tracer self-diffusion coefficients $D_{s}^{\rm{tr}}$ and the shear viscosity $\eta_{s}$ of individual species $s$.

For the microscale MD, the diffusion $D_{s}^{\rm{tr}}$ and shear viscosity $\eta_{s}$ are defined by 
\begin{equation}
\begin{aligned}
    D_{s}^{\rm{tr,MD}} &= \frac{1}{3N_{s}}\sum_{i\in s} \int_{0}^{\infty} \left\langle \bm{c}_{i}(t) \cdot \bm{c}_{i}(0) \right\rangle \dd t, \\
	\eta_{st}^{\rm{MD}} &= \frac{V}{3k_{B}T} \int_{0}^{\infty}  \sum_{\alpha \beta \in \{xy,xz,yz\}} \left\langle P_{s,\alpha\beta}^{\rm{kin}}(t) P_{t, \alpha\beta}^{\rm{kin}}(0) \right\rangle \dd t, \\
	P_{s, \alpha\beta}^{\rm{kin}} &= \frac{1}{V}\sum_{i\in s} m_{s} c_{i,\alpha} c_{i,\beta}, \quad 
    \boldsymbol{\eta}^{\rm MD}=[\eta_{st}^{\rm MD}]_{s,t=1}^{S},
\end{aligned}
\end{equation}
where $\bm{\eta}_{\rm sp}^{\rm MD} = \boldsymbol{\eta}^{\rm MD}\bm{1}$ represents the vector-valued shear viscosity for individual species $\eta_s^{\rm MD}$, $s = 1, \cdots, S$.

For the kinetic model \cref{eq:multi-kinetic}, the transport coefficients can be derived from the perturbation of the density $\eta_s$ near the equilibrium state. 
At the common Maxwellian $M_{s}$, write the perturbation density $f_{s}=n_{s}M_{s}(1+h_{s})$ and the evolution of the perturbation $\bm{h}$ as $\partial_{t}\bm{h}=-\mathcal{A}\bm{h}$.
The positive linearized operator is characterized by
\begin{equation}\label{eq:transport-dirichlet-short}
\begin{aligned}
	\langle \bm{h}, \mathcal{A} \bm{k} \rangle_{\bm{n}M} =& \frac{1}{2} \sum_{s,t=1}^{S}n_{s}n_{t} \iint M_{s}M_{t}' \bm{D}_{st}[h]^{T} \bm{\omega}_{st} \bm{D}_{st}[k] \dd\bm{v}\dd\bm{v}',
    \\ 
    (\mathcal A \bm{h})_{s} =& - \dfrac{1}{m_{s} M_{s}} \sum_{t = 1}^{S} n_{t} \nabla_{\bm{v}}\!\cdot\!\int M_{s} M_{t}' \bm{\omega}_{st} \bm{D}_{st}[h] \dd \bm{v}'.
\end{aligned}
\end{equation}
where $\langle \bm{h}, \bm{k} \rangle_{\bm{n}M} = \sum_{s=1}^{S}n_{s} \int M_{s} h_{s} k_{s}\dd\bm{v}$ and $\bm{D}_{st}[h] = \nabla h_{s}/m_{s} - \nabla' h_{t}'/m_{t}$.

For tracer diffusion of species $s$, we introduce a dilute tagged population $s^{*}$
\begin{equation}
	f_{s^{*}} = n_{s^{*}}(\bm{x}) M_{s} \left[ 1 - \chi_{s}^{\mathrm{tr}}(\bm{c}) \partial_{x}\log n_{s^{*}} \right], \quad n_{s^{*}}\ll n_{s},
\end{equation}
while keeping all bath Maxwellians fixed.
The first-order tagged kinetic equation and the Fick law give the diffusion coefficient
\begin{equation}\label{eq:tracer-short}
\begin{aligned}
	\mathcal{A}_{\mathrm{tr},s} \chi_{s}^{\mathrm{tr}} =& c_{x}, \quad 
    \mathcal A_{{\rm tr}, s} h = - \dfrac{1}{m_{s}^{2} M_{s}} \sum_{t = 1}^{S} n_{t} \nabla_{\bm{v}} \cdot \int M_{t}' M_{s} \bm{\omega}_{st} \nabla h \dd \bm{v}' , \\
    J_{s^{*}, x} =& \int c_{x} f_{s^{*}} \dd \bm{v} = -D_{s}^{\mathrm{tr}} \partial_{x}n_{s^{*}}, \quad
    D_{s}^{\mathrm{tr}} = \langle c_{x}, \mathcal{A}_{\mathrm{tr},s}^{-1} c_{x} \rangle_{M_{s}}.
\end{aligned}
\end{equation}
To determine $\chi_{s}^{\mathrm{tr}}$, we use an approximation up to the second-order Sonine basis, i.e., 
\[
\chi_{s}^{\mathrm{tr}}=\sum_{j=0}^{1} a_{sj}^{D}\phi_{sj}^{D}, \quad \phi_{s0}^{D} = c_{x}, \quad 
	\phi_{s1}^{D} = \left( 5/2 - |\bm{c}|^{2}/2 \sigma_{s}^{2} \right) c_{x},
\]
where $\sigma_{s}^{2} = k_{B}T/m_{s}$. 
The coefficients $a_{sj}^{D}$ are determined by Galerkin projection 
\begin{equation}\label{eq:tracer-sonine-short}
\begin{aligned}
	\sum_{j=0}^{1} \widehat{C}_{s,ij}^{D} a_{sj}^{D} &= \widehat{b}_{s,i}^{D}, \\
	\widehat{C}_{s,ij}^{D} &= \frac{1}{m_{s}^{2}} \sum_{t=1}^{S} n_{t} \left\langle (\nabla \phi_{si}^{D})^{T} \bm{\omega}_{st} \nabla\phi_{sj}^{D} \right\rangle_{M_{s}M_{t}'},  \\
	\widehat{b}_{s,i}^{D} &= \int M_{s} \phi_{si}^{D}c_{x} \dd \bm{v},
\end{aligned}
\end{equation}
and therefore the tracer self-diffusion coefficient $D_{s}^{\mathrm{tr},(2)}$ is given by 
\begin{equation}
D_{s}^{\mathrm{tr},(2)} = (\widehat{\bm{b}}_{s}^{D})^{T} (\widehat{\bm{C}}_{s}^{D})^{-1} \widehat{\bm{b}}_{s}^{D},  
\label{eq:diffusion_kinetic}
\end{equation}
where the superscript $2$ represents the second-order Sonine basis approximation. 

Similarly, for shear viscosity $\eta_s$, we introduce perturbation $f_{s} = n_{s} M_{s} \left( 1 - \chi_{s}^{\eta} \dot{\gamma} \right)$, where $\dot{\gamma} $ is the shear rate. 
The first-order kinetic equation and Newton's law give
\begin{equation}\label{eq:viscosity-short}
	\mathcal{A}\bm{\chi}^{\eta} = \bm{G}^{\eta}, \quad 
    \Pi_{xy} = \sum_{s} m_{s} \int c_{x} c_{y} f_{s} \dd \bm{v} = -\eta \dot{\gamma}, \quad
	\eta = k_{B} T \langle \bm{G}^{\eta}, \mathcal{A}^{\dagger} \bm{G}^{\eta} \rangle_{\bm{n}M}.
\end{equation}
where $G_{s}^{\eta} = \frac{m_{s}c_{x}c_{y}}{k_{B}T}$. To determine $\bm{\chi}^{\eta}$, we use the two-term Sonine basis, i.e.,
\begin{equation}
    \chi^{\eta}_s=\sum_{j=0}^{1}  a_{sj}^{\eta}\phi_{sj}^{\eta}, \quad
	\phi_{s0}^{\eta}=c_{x}c_{y}, \quad
	\phi_{s1}^{\eta} = \left( 7/2 - |\bm{c}|^{2}/2\sigma_{s}^{2} \right)c_{x}c_{y}.
\end{equation}
The coupled $2S \times 2S$ Galerkin system is $\bm{C}^{\eta} \bm{a}^{\eta} = \bm{b}^{\eta}$, with
\begin{equation}\label{eq:viscosity-galerkin-short}
\begin{aligned}
	C_{si,tj}^{\eta} &= \frac{1}{2} \sum_{a,b=1}^{S} n_{a}n_{b} \left\langle (\bm{d}_{ab}^{si,\eta})^{T} \bm{\omega}_{ab} \bm{d}_{ab}^{tj,\eta} \right\rangle_{M_{a}M_{b}'}, \\
	b_{si}^{\eta} &= n_{s} \int M_{s} \phi_{si}^{\eta} G_{s}^{\eta} \dd \bm{v}, \quad
	\bm d_{ab}^{si,\eta} = \frac{\delta_{as}}{m_{a}} \nabla \phi_{ai}^{\eta} - \frac{\delta_{bs}}{m_{b}} \nabla' \phi_{bi}^{\eta\prime}.
\end{aligned}
\end{equation}

Partitioning by Sonine order and using $\bm{b}_{1}^{\eta}=\bm{0}$ gives
\begin{equation}\label{eq:viscosity-sonine-short}
\begin{aligned}
	\bm{C}_{\rm{eff}}^{\eta} &= \bm{C}_{00}^{\eta} - \bm{C}_{01}^{\eta}(\bm{C}_{11}^{\eta})^{-1} \bm{C}_{10}^{\eta}, \quad \quad
	\bm{B}_{0}^{\eta} = \operatorname{diag}(b_{10}^{\eta}, \ldots, b_{S0}^{\eta}), \\
	\boldsymbol{\eta}^{(2)} &= k_{B} T(\bm{B}_{0}^{\eta})^{T} (\bm{C}_{\rm eff}^{\eta})^{\dagger} \bm{B}_{0}^{\eta}, \quad 	\bm{\eta}_{\rm{sp}}^{(2)} = \boldsymbol{\eta}^{(2)}\bm{1} 
\end{aligned}
\end{equation}
where $\bm{\eta}_{\rm{sp}}^{(2)}$ represents the species-contribution vectors of shear viscosity $\eta_{s}^{(2)}$, $s = 1, \cdots, S$. 
As shown in \Cref{sec:numerical_transport_coeff}, the numerical comparison between $D_{s}^{\rm{tr,MD}}$ and $D_{s}^{\rm{tr,(2)}}$ as well as $\eta_{s}^{\rm{MD}}$ and $\eta_{s}^{\rm{(2)}}$ provides a validation of the constructed collision operator.

\section{Numerical results}\label{sec:numerical}

\subsection{Identical-species test and Landau limit}\label{subsec:1sp_relax_time}

We first test consistency of one-species plasma kinetics under an artificial splitting into two labels. Let $a$ and $b$ representing physically identical particles with different number densities $n_{a}$ and $n_{b}$, the same mass $m$, charge $q$, and  initial normalized PDF $\eta_{0}(\bm{v})$.
If all collision-kernel blocks are identical $\bm{\omega}$, the two-component model reduces to
\begin{equation}\label{eq:reduced-identical-equation}
    \partial_{t} \eta = (n_a+n_b) \mathcal{C}_{\bm{\omega}}[\eta,\eta] = n \mathcal{C}_{\bm{\omega}}[\eta,\eta], \quad 
    n=n_a+n_b,
\end{equation}
where $\mathcal{C}_{\omega}$ are represented by the present DDCO \cref{eq:two-encoder-kernel} \cref{eq:separation} and the Landau model \cref{eq:landau_model}. 
For the identical-species case, the Landau collision kernel takes a simpler form
\[
\bm{\omega}(\bm c', \bm c) \propto \mathcal{P} |\bm{u}|^{-1} \ln\Lambda, \quad  \ln\Lambda =\ln \left[ \frac{2}{\sqrt{3} \Gamma^{3/2}} \right], \quad \Gamma= \frac{q^2}{4\pi\epsilon_{0} k_{B}T} \left(\frac{4\pi n}{3}\right)^{1/3},
\]
where $\Gamma$ is the dimensionless coupling parameter and the Landau model remains valid for the weakly coupled regime $\Gamma \ll 1$ or equivalently $\Lambda \gg 1$. In this study, we choose the physical parameters such that $\Gamma$ across the weakly coupled regime $\mathcal{O}(10^{-1})$ and the moderately coupled regime $\mathcal{O}(1)$. 

We choose bi-Maxwellian with $T_{y}(0)=T_{z}(0)=4T_{x}(0)$ as the initial state for the kinetic model \cref{eq:reduced-identical-equation} and the full MD simulations. We examine two non-linear measures of the same anisotropy \cite{ichimaru1970relaxation,hellinger2009coulomb},
\begin{equation}
\nonumber
    A_{K}=\frac{T_{\perp}}{T_{\parallel}}-1, \quad
    A_{T}=\frac{T_{\perp}-T_{\parallel}}{T} = \frac{A_K}{1+2A_K/3}.
\end{equation}
where $T_{\parallel}=T_{x}$, $T_{\perp}=\frac{T_{y}+T_{z}}{2}$. For this benchmark, the Landau model predicts the short-time exponential decay of $A_K$ and $A_T$ following the Landau-Kogan moment equation
\cite{ichimaru1970relaxation,baalrud2017temperature}, and both relaxation times are scaled by
$
\tau_{\rm{L}}^{\rm{A}} \sim \frac{T^{3/2}}{n\ln\Lambda}$, i.e.,
\begin{equation}
\left.\frac{\partial \ln \tau_{\rm{L}}^{\rm{A}}}{\partial \ln T}\right|_{n} \sim \frac{3}{2} - \frac{3}{2\ln \Lambda}, \qquad \left.\frac{\partial \ln \tau_{\rm{L}}^{\rm{A}}}{\partial \ln n}\right|_{T} \sim -1 + \frac{1}{2\ln \Lambda},
\label{eq:tau_scaling}
\end{equation}
which take $1.5$ and $-1$ for the Landau limit $\Gamma \ll 1$.

\Cref{fig:AT,fig:AK} show the relaxation time predicted from the kinetic model \cref{eq:reduced-identical-equation} and the full MD simulations. In the weakly coupled regime $\Gamma = 0.1$ and $\ln \Lambda = 3.6$, the predictions of both the Landau and the present DDCO model show good agreement with the MD result, and the relaxation time $\tau^A$ scales $T$ and $n$ by $1.09$ and $-0.86$ following Eq. \cref{eq:tau_scaling}. 
As $\Gamma \ge 1$, the Landau model is inaccurate and $\ln \Lambda$ is not well-defined. The full MD results show that $\tau^A$ scales with $T$ by $0.55$ and $n$ by $-0.68$ for $A_T$, and scales with $T$ by $0.3$ and $n$ by $-0.6$ for $A_K$, respectively. 
For both cases, the prediction from the present DDCO model shows good agreement with the MD results for this moderately coupled regime.

\begin{figure}[htbp]
    \centering
    \begin{minipage}{0.32\linewidth}
    \centering
    \includegraphics[width=\linewidth]{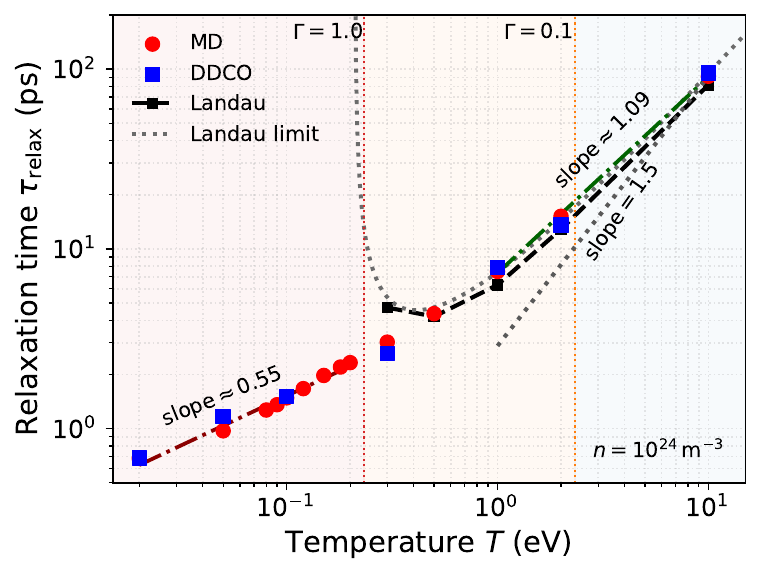}
    \end{minipage}\hfill
    \begin{minipage}{0.32\linewidth}
    \centering
    \includegraphics[width=\linewidth]{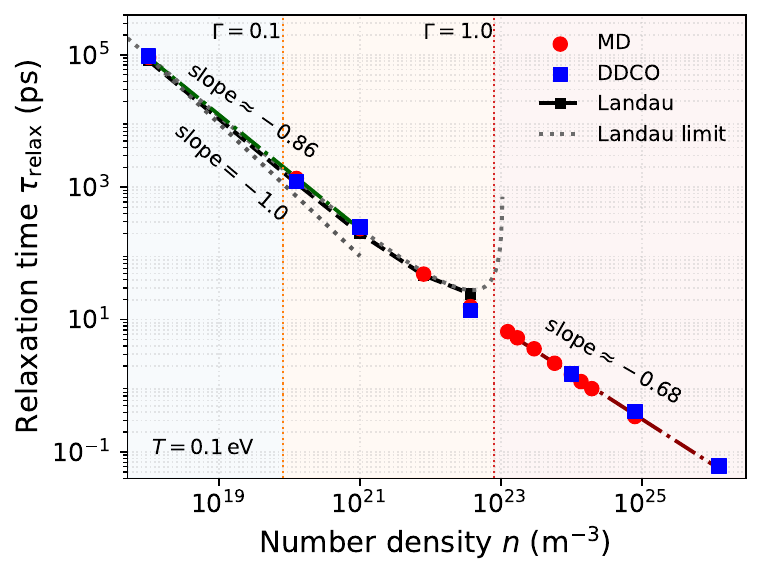}
    \end{minipage}\hfill
    \begin{minipage}{0.32\linewidth}
    \centering
    \includegraphics[width=\linewidth]{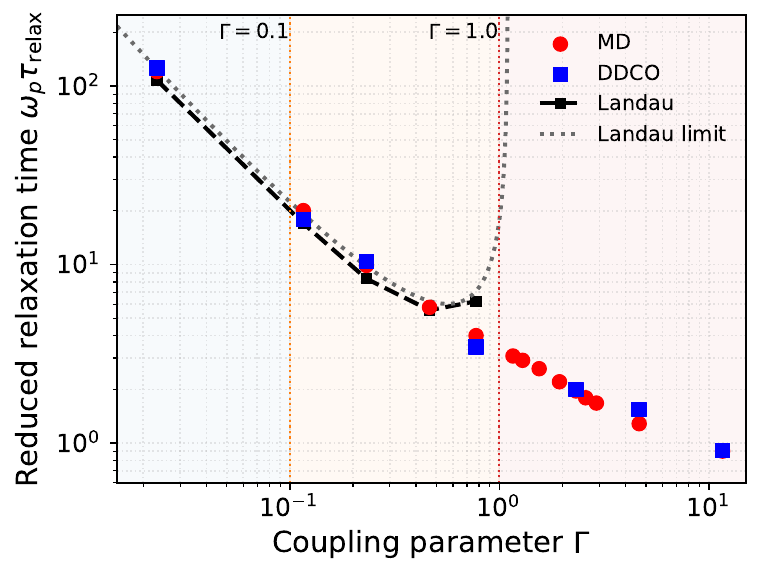}
    \end{minipage}
    \caption{Effective relaxation time of $A_{T}$ versus temperature at $n=10^{24}\,\mathrm{m}^{-3}$ (left), number density at $T=0.1\,\mathrm{eV}$ (center), and coupling parameter (right).
    Red circles and blue squares denote MD and DDCO, respectively, and the black dash-dotted curves are the finite-anisotropy and small-$A$ Landau predictions.
    Background colors separate the weak, moderate, and strong coupling regimes.}
    \label{fig:AT}
\end{figure}

\begin{figure}[htbp]
    \centering
    \begin{minipage}{0.32\linewidth}
    \centering
    \includegraphics[width=\linewidth]{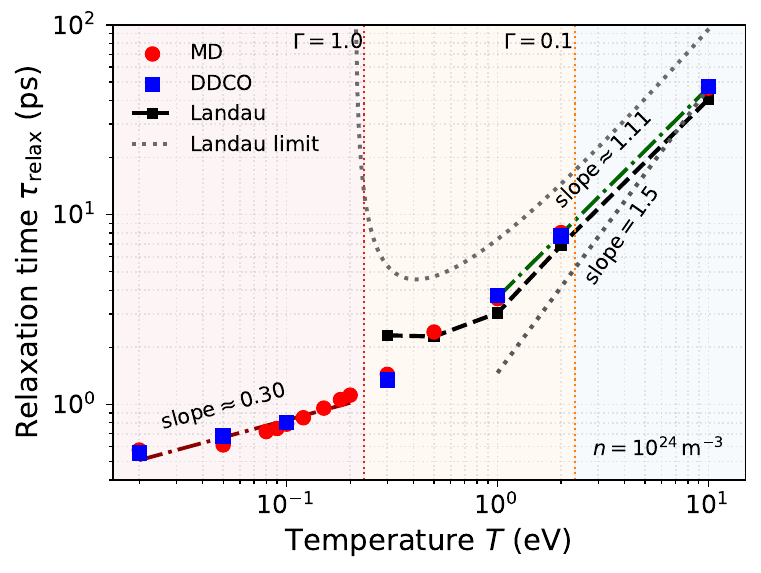}
    \end{minipage}\hfill
    \begin{minipage}{0.32\linewidth}
    \centering
    \includegraphics[width=\linewidth]{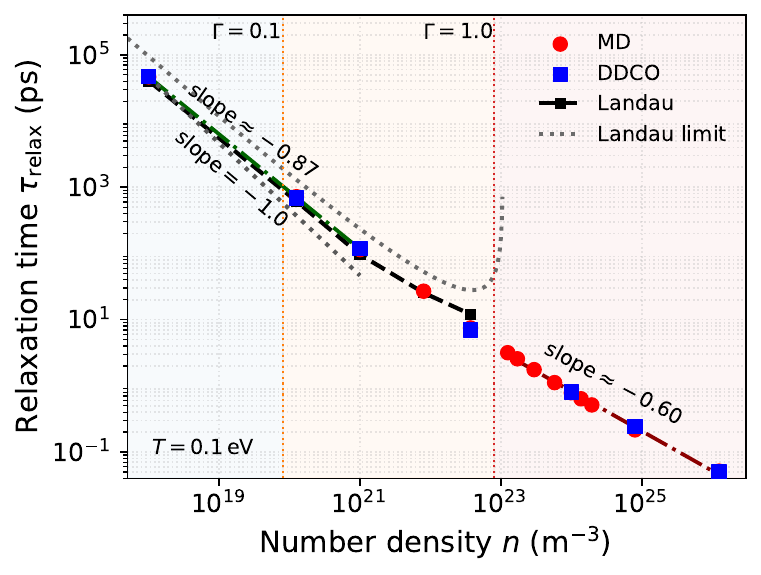}
    \end{minipage}\hfill
    \begin{minipage}{0.32\linewidth}
    \centering
    \includegraphics[width=\linewidth]{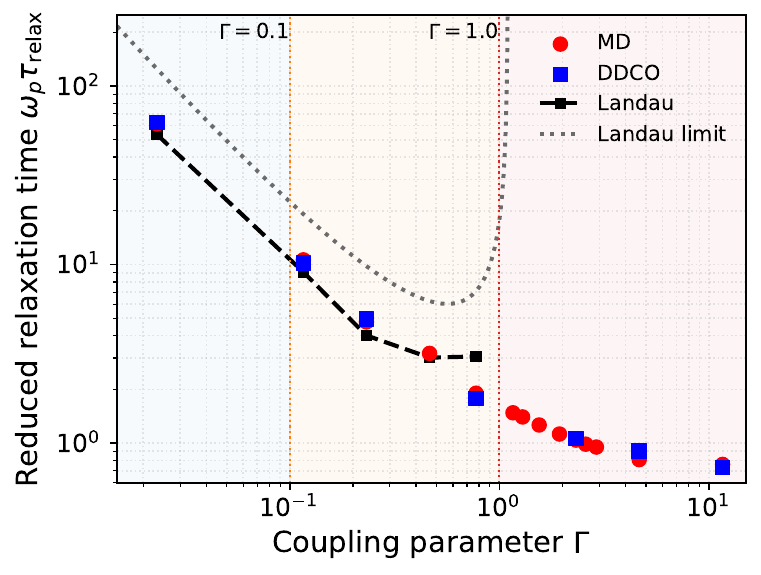}
    \end{minipage}
    \caption{Effective relaxation time of $A_{K}$ versus temperature (left), number density (center), and coupling parameter (right), with the notation and states of \Cref{fig:AT}.}
    \label{fig:AK}
\end{figure}

\subsection{Transport coefficients of two-species (Hydrogen-Tritium) and three-species (Hydrogen-Deuterium-Tritium) systems}
\label{sec:numerical_transport_coeff}
We next compare the predictions of the diffusion and shear viscosity of the present multi-species collision operator by \cref{eq:tracer-sonine-short,eq:viscosity-sonine-short} with the corresponding MD results.
All species are singly charged, $q_{s}=e$, and all three cases have $T=0.2\,\mathrm{eV}$.
Case 1 uses $(m_{a}, m_{b})=(m_{p}, 3m_{p})$ and $(n_{a},n_{b})=(1\times10^{24}\,\mathrm{m}^{-3}, 3\times10^{24}\,\mathrm{m}^{-3})$, case 2 reverses the two densities, and case 3 uses $(m_{a},m_{b},m_{c})=(m_{p}, 2m_{p}, 3m_{p})$ and $n_{a}=n_{b}=n_{c}=10^{24}\,\mathrm{m}^{-3}$.
For any reported quantity $X$, we define the relative error $\varepsilon_{X} = |X^{\rm DDCO}-X^{\rm MD}|/|X^{\rm MD}|$.

\begin{table}[htbp] 
	\caption{Species tracer diffusion coefficients predicted from MD and the DDCO model.}
    \label{tab:tracer-comparison}
    \centering
	\begin{tabular}{ccccc}
	\toprule
	Case & Species & $D_{s}^{\rm tr,MD}$ ($\mathrm{m}^{2}/\mathrm{s}$) & $D_{s}^{\rm tr,(2)}$ ($\mathrm{m}^{2}/\mathrm{s}$) & $\varepsilon_{D_s}$ \\
	\midrule
	1 & $a$ & $2.652\times10^{-5}$ & $2.990\times10^{-5}$ & $12.7\%$ \\ 
	1 & $b$ & $2.252\times10^{-5}$ & $2.479\times10^{-5}$ & $10.1\%$ \\ 
	2 & $a$ & $3.133\times10^{-5}$ & $3.217\times10^{-5}$ & $2.7\%$ \\ 
	2 & $b$ & $2.892\times10^{-5}$ & $3.061\times10^{-5}$ & $5.8\%$ \\ 
	3 & $a$ & $3.386\times10^{-5}$ & $3.195\times10^{-5}$ & $5.6\%$ \\ 
	3 & $b$ & $3.060\times10^{-5}$ & $3.389\times10^{-5}$ & $10.8\%$ \\ 
	3 & $c$ & $2.988\times10^{-5}$ & $2.742\times10^{-5}$ & $8.2\%$ \\ 
	\bottomrule
	\end{tabular}
\end{table}

\begin{table}[htbp] 
	\caption{Species contributions to the shear viscosity $\bm\eta_{\rm sp}=\boldsymbol\eta\bm 1$ predicted from MD and the DDCO model.}
    \label{tab:viscosity-comparison}
    \centering
	\begin{tabular}{ccccc}
	\toprule
	Case & Species & $\eta_{{\rm sp},s}^{\rm MD}$ ($\mathrm{Pa}\,\mathrm{s}$) & $\eta_{{\rm sp},s}^{\rm kin,(2)}$ ($\mathrm{Pa}\,\mathrm{s}$) & $\varepsilon_{\eta_s}$ \\
	\midrule
	1 & $a$ & $3.436\times10^{-8}$ & $3.266\times10^{-8}$ & $4.9\%$ \\ 
	1 & $b$ & $1.756\times10^{-7}$ & $1.712\times10^{-7}$ & $2.5\%$ \\ 
	2 & $a$ & $9.377\times10^{-8}$ & $9.772\times10^{-8}$ & $4.2\%$ \\ 
	2 & $b$ & $7.372\times10^{-8}$ & $8.247\times10^{-8}$ & $11.9\%$ \\ 
	3 & $a$ & $4.542\times10^{-8}$ & $4.037\times10^{-8}$ & $8.9\%$ \\ 
	3 & $b$ & $5.892\times10^{-8}$ & $6.271\times10^{-8}$ & $6.4\%$ \\ 
	3 & $c$ & $8.128\times10^{-8}$ & $8.802\times10^{-8}$ & $8.3\%$ \\ 
	\bottomrule
	\end{tabular}
\end{table}

We emphasize that the transport coefficients have not been used in the relaxation loss function in Eq. \cref{eq:loss}. Rather, these comparisons examine the linear response of the collision operator with respect to various perturbation modes near equilibrium. 
As shown in \Cref{tab:tracer-comparison,tab:viscosity-comparison}, the predictions of the present DDCO model show good agreement with the MD results, with errors range from about $2.5\%$ to $12.7\%$. In particular, the accurate agreement for the three-species system indicates that the present DDCO model achieves accurate modeling of the inter-species relaxation processes.

\subsection{Temperature relaxation of bi-Maxwellian distributions}

With the three cases defined in \Cref{sec:numerical_transport_coeff}, we examine the kinetic relaxation by setting the initial temperature of each species as $T_{y}(0) = T_{z}(0) = 4T_{x}(0)$ while fixing the scalar temperature $T=(T_{x}+T_{y}+T_{z})/3 = 0.2\,\mathrm{eV}$.
The DDCO accurately predicts both the species-dependent relaxation rates and the equilibrium temperature.
In particular, the swap of the number density ratio $n_1:n_2$ between cases 1 and 2 changes the relative relaxation of species $a$ and $b$.
In case 3, the relaxation becomes progressively slower from the lighter to the heavier component, while all components approach the same energy-constrained equilibrium.
These results verify that the constructed kernel $\bm\omega_{st}$ accurately captures both the intra- and inter-species energy exchange.

\begin{figure}[htbp]
    \centering
    \includegraphics[width=0.48\linewidth]{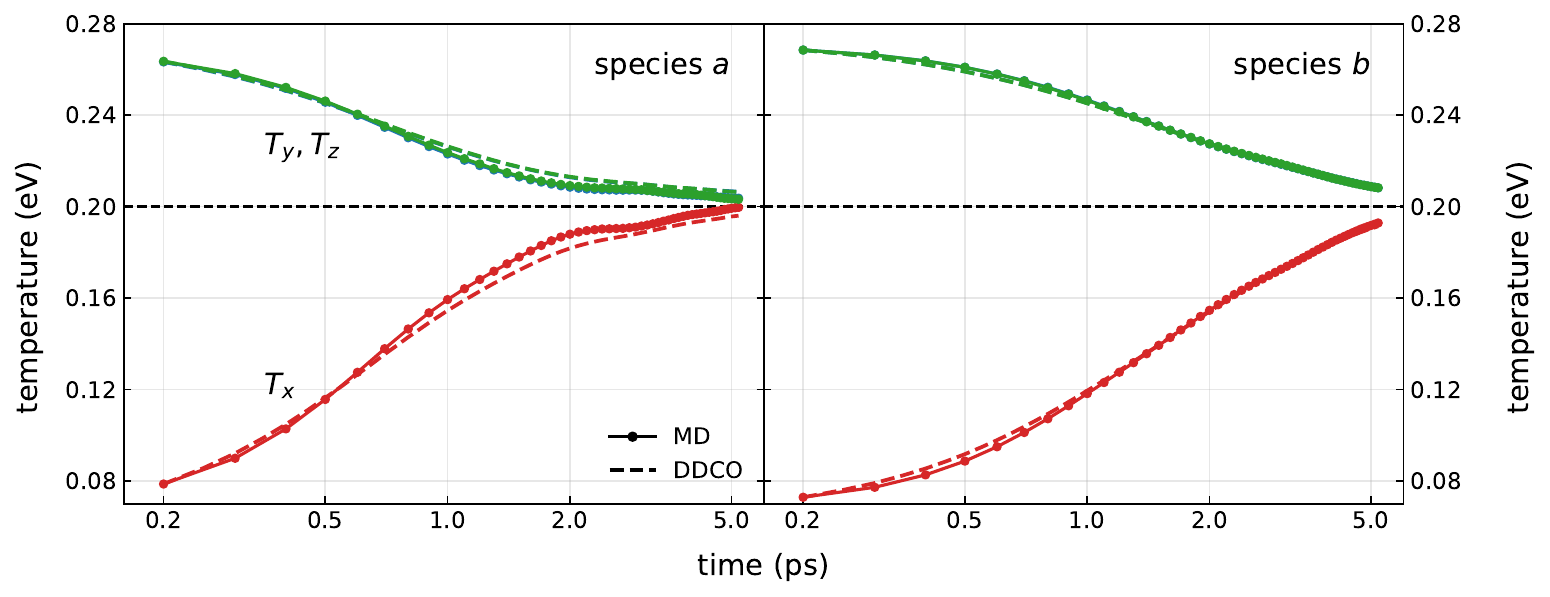} 
    \includegraphics[width=0.48\linewidth]{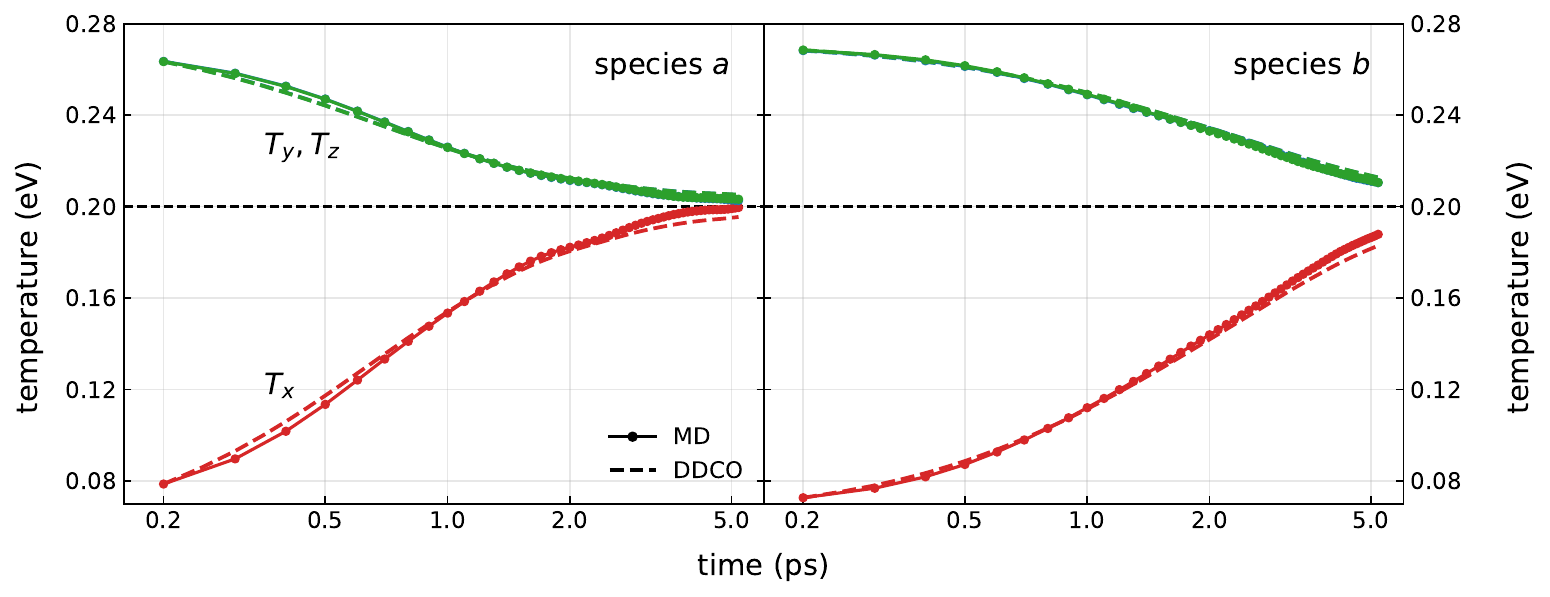} \\ 
    \includegraphics[width=0.72\linewidth]{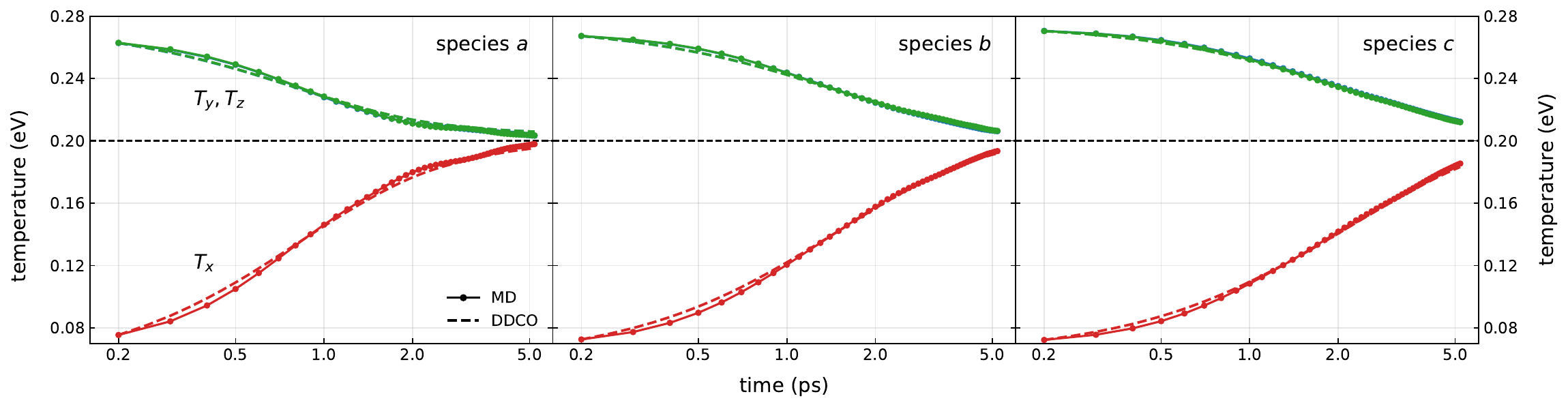}
    \caption{Component-temperature relaxation for case 1 (upper left), case 2 (upper right), and case 3 (bottom) with solid curves for MD results and dashed curves for the DDCO model.
    Red denotes $T_{x}$, green and blue denote $T_{y}$ and $T_{z}$, and the horizontal dashed line is the energy-determined equilibrium temperature $0.2\,\mathrm{eV}$.}
    \label{fig:temp}
\end{figure}

\subsection{Instantaneous velocity distribution}
Beyond the temperature relaxation, we further examine the kinetic process by comparing the instantaneous velocity distribution predicted from both the MD and DDCO model.  
In particular, we compare the kinetic process with four initial distributions: a symmetric double-well, an asymmetric double-well, a diagonal bimodal state, and a trimodal state, see in \Cref{app:md-learning}.

\begin{figure}[htbp]
    \centering
    \includegraphics[width=0.48\linewidth]{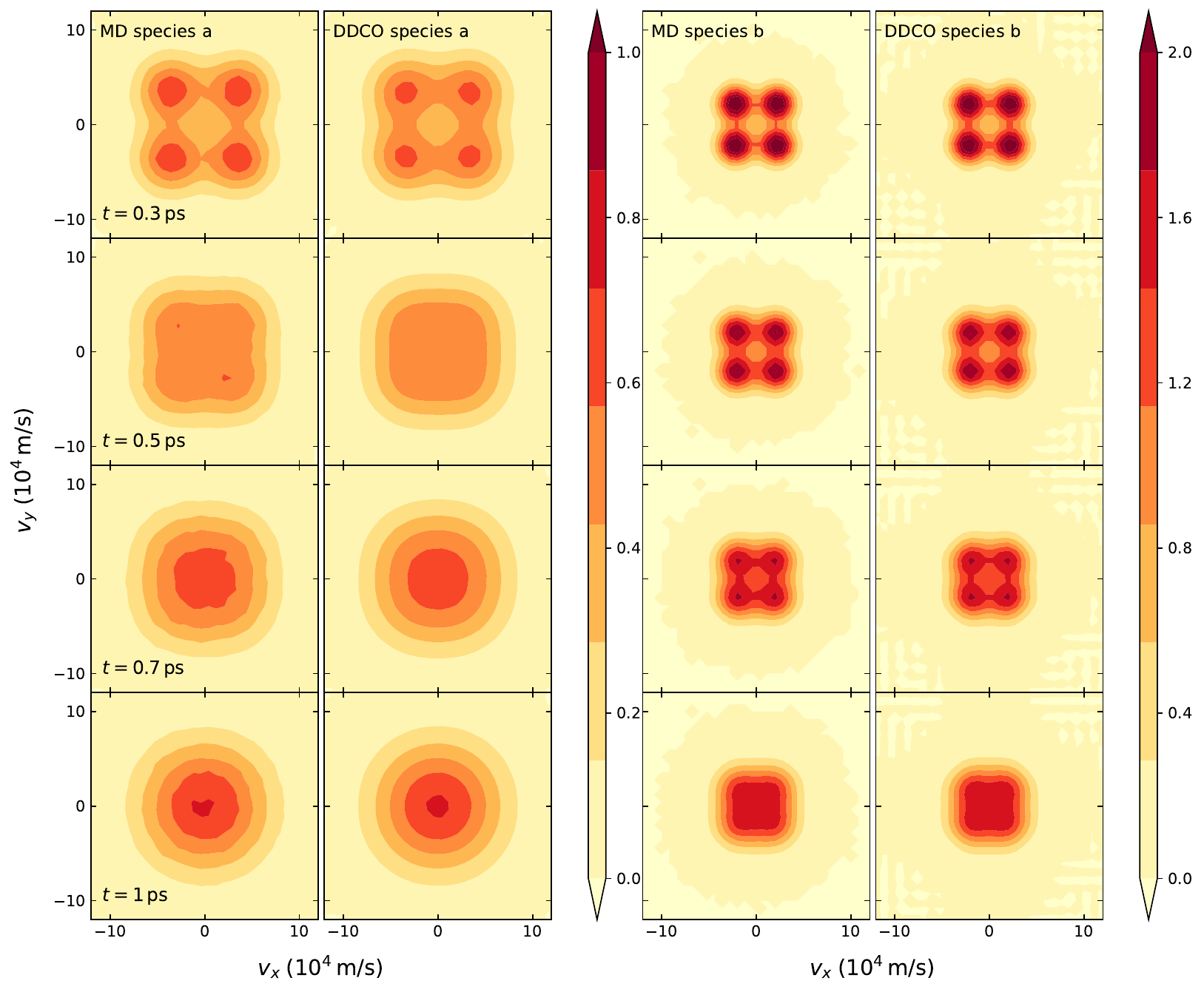} 
    \includegraphics[width=0.48\linewidth]{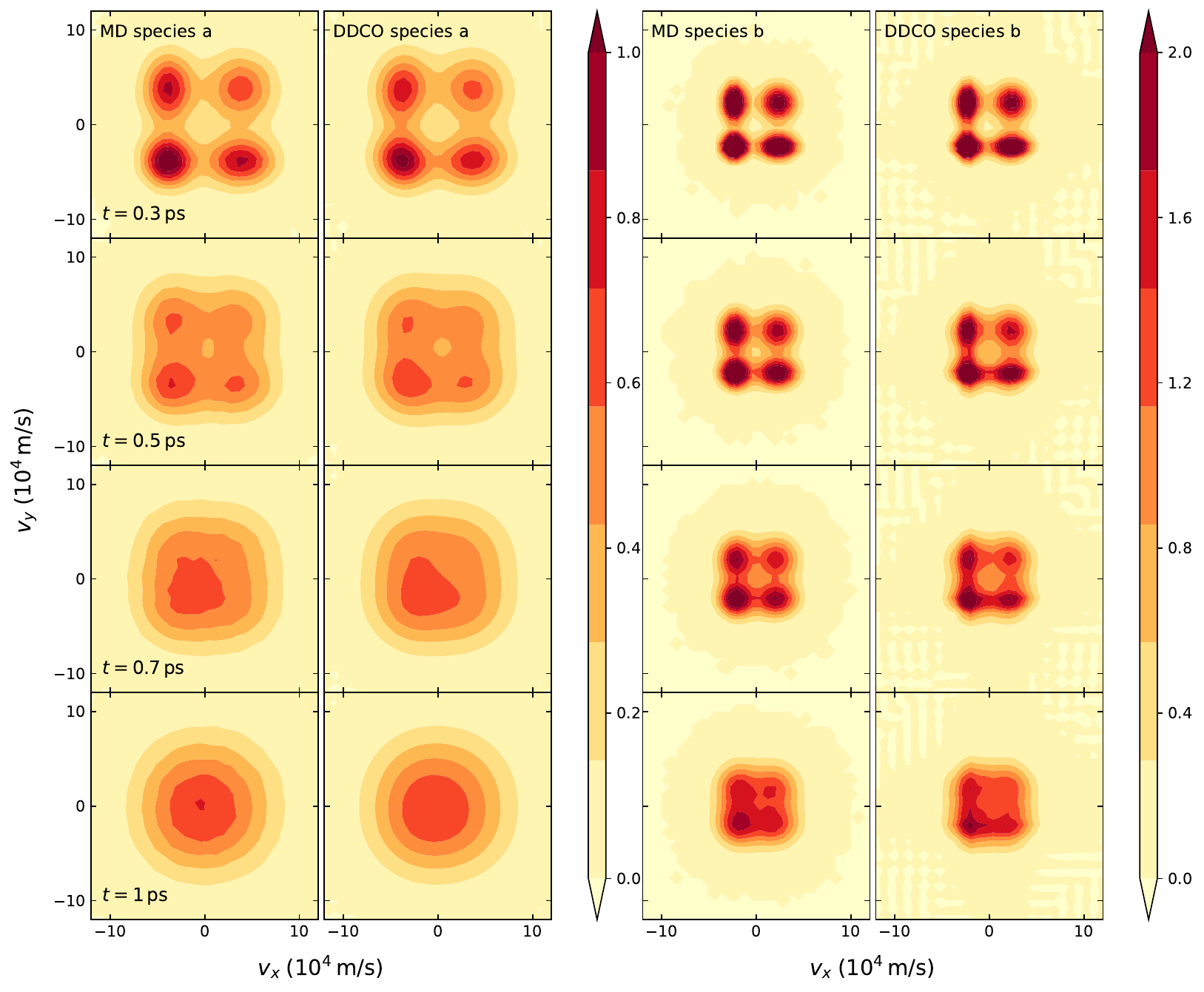} \\ 
    \centering
    \includegraphics[width=0.48\linewidth]{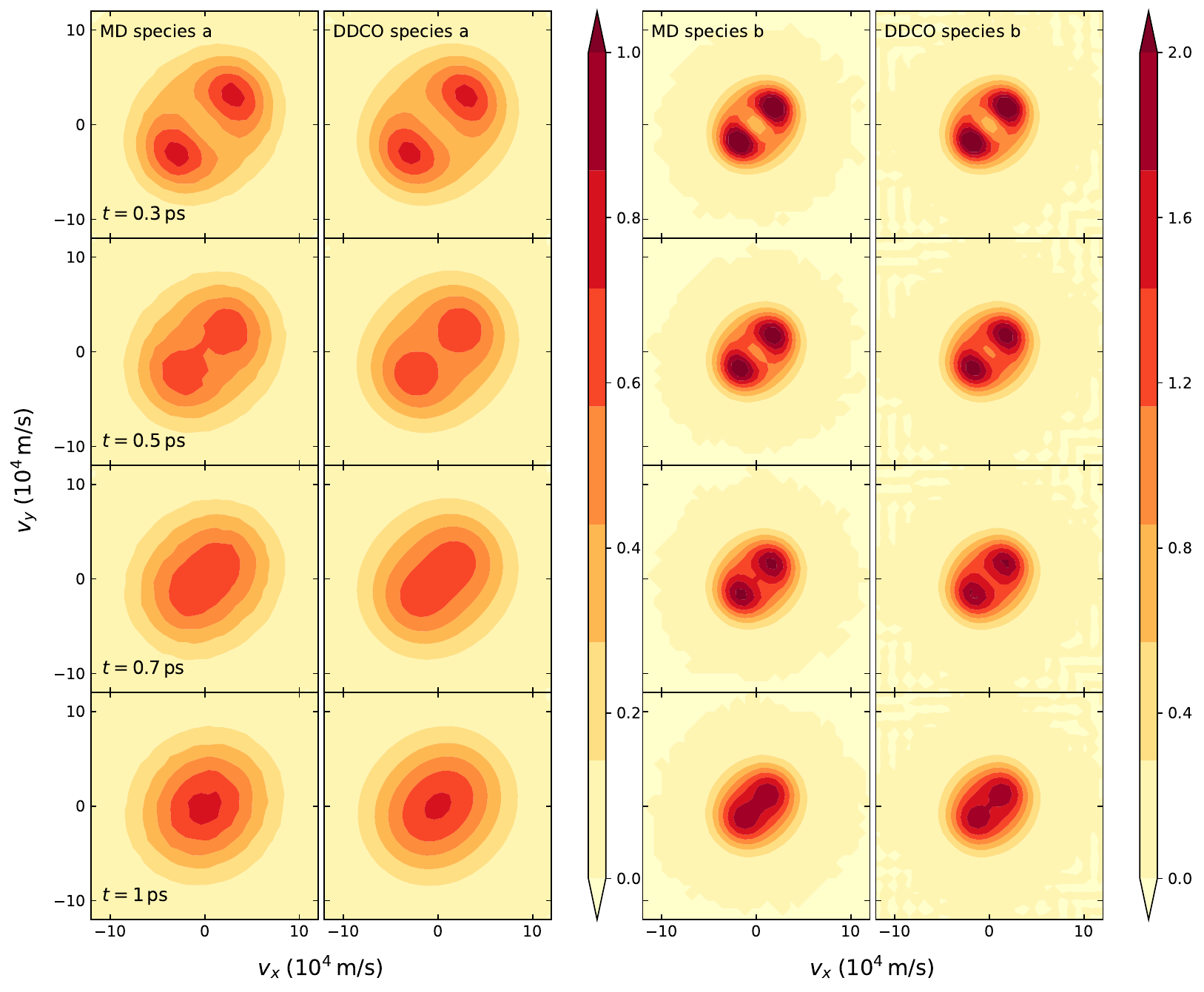}
    \includegraphics[width=0.48\linewidth]{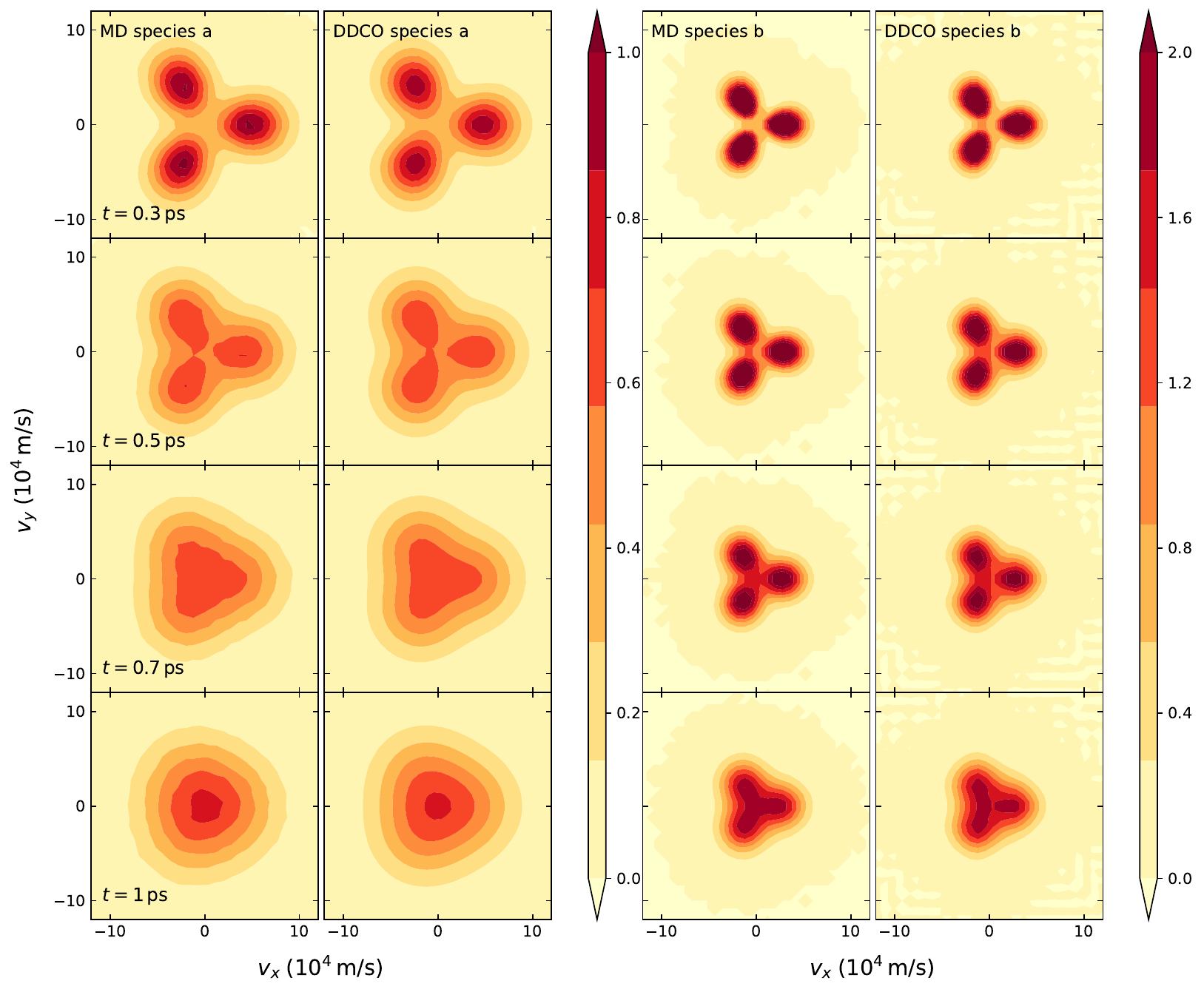}
    \caption{Comparison of the instantaneous velocity distribution on the $v_x$-$v_y$ plane for the symmetric double-well (upper left), asymmetric double-well (upper right), diagonal bimodal state (lower left), and trimodal state (lower right) in case 1.
    Within every panel, rows correspond to $t=0.3$, $0.5$, $0.7$, and $1\,\mathrm{ps}$, and the MD and DDCO columns are shown side by side for each species.}
    \label{fig:mdf_1_3}
\end{figure}

\begin{figure}[htbp]
    \centering
    \includegraphics[width=0.48\linewidth]{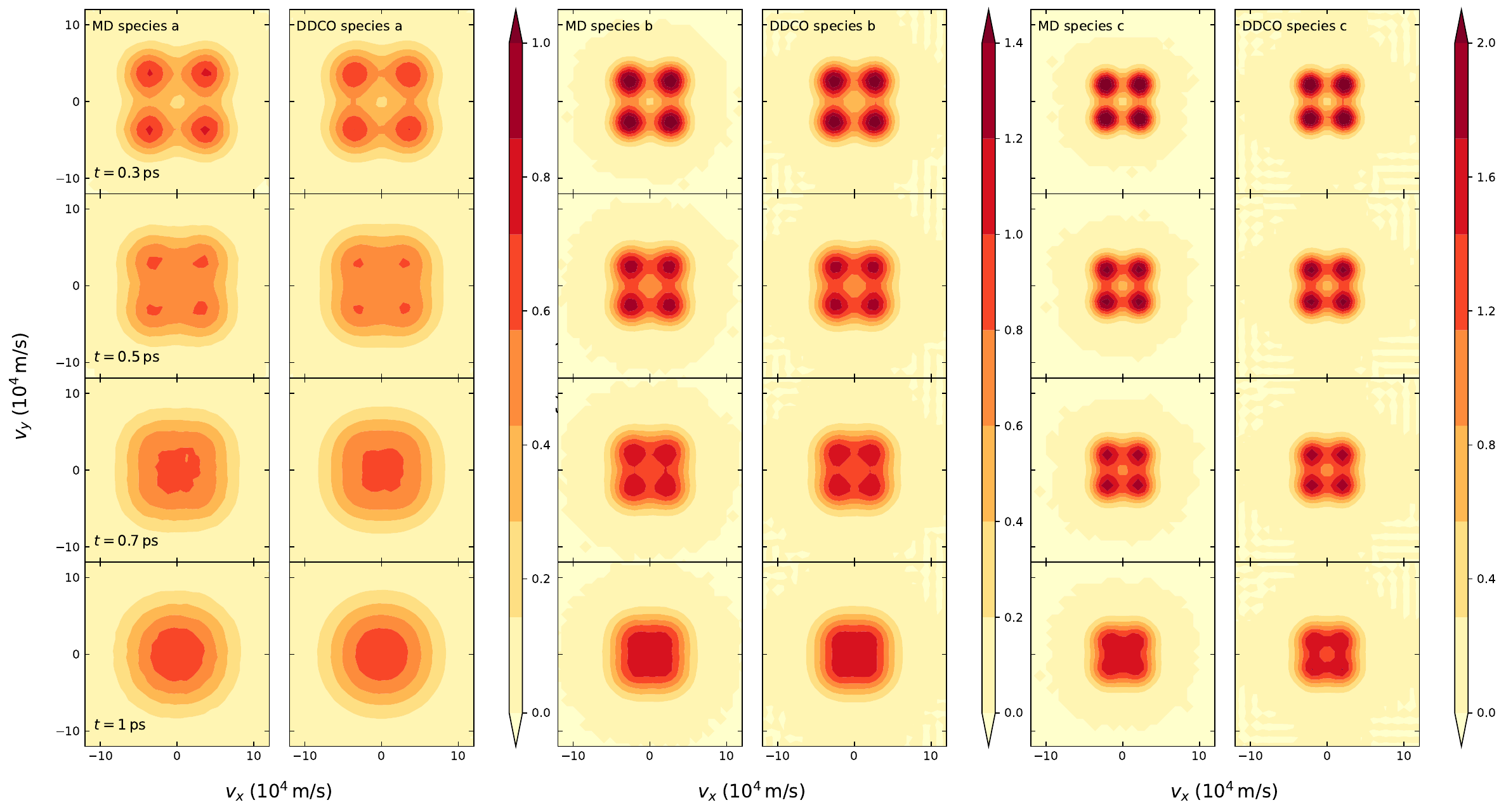} 
    \includegraphics[width=0.48\linewidth]{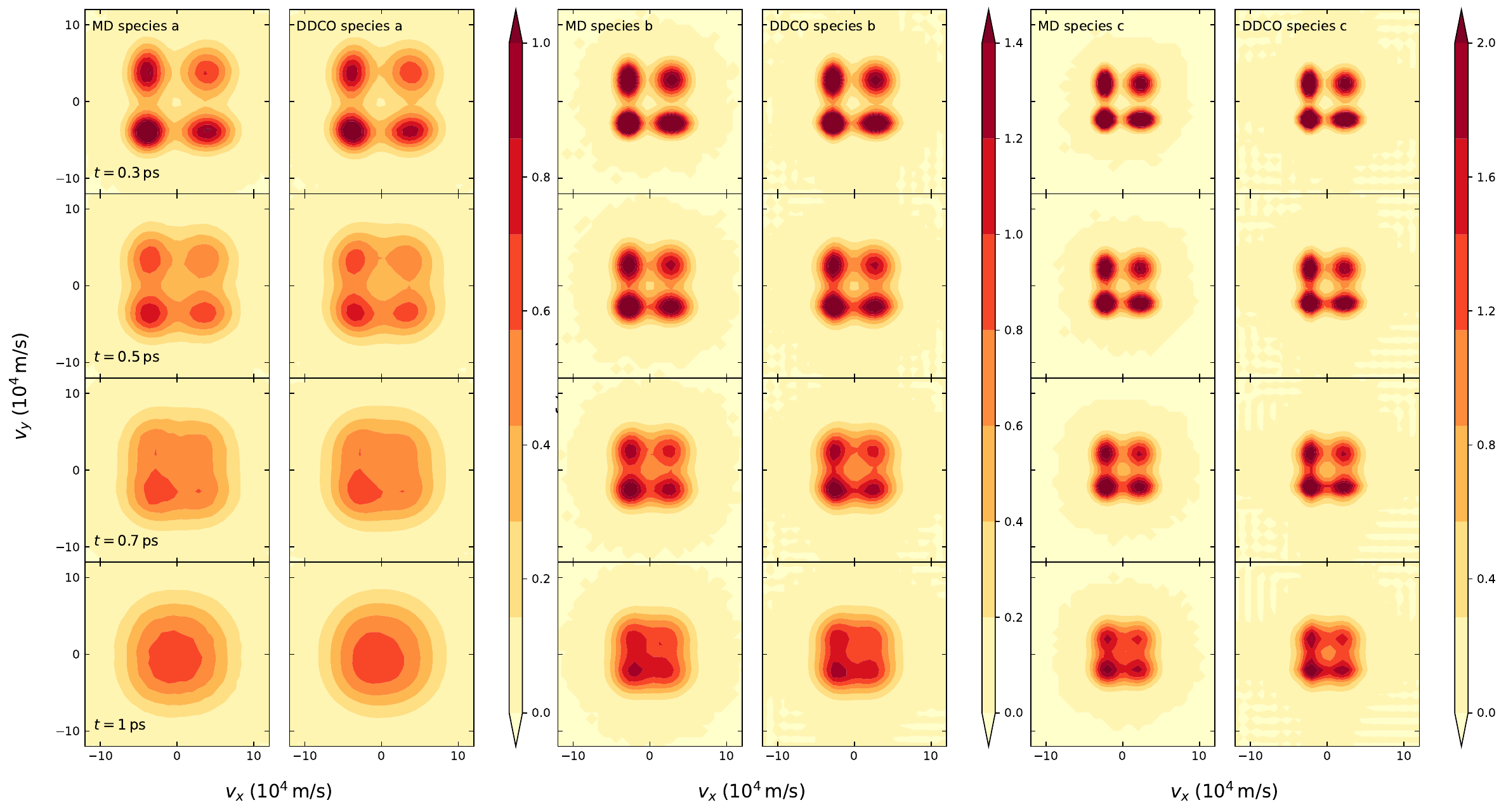} \\ 
    \centering
    \includegraphics[width=0.48\linewidth]{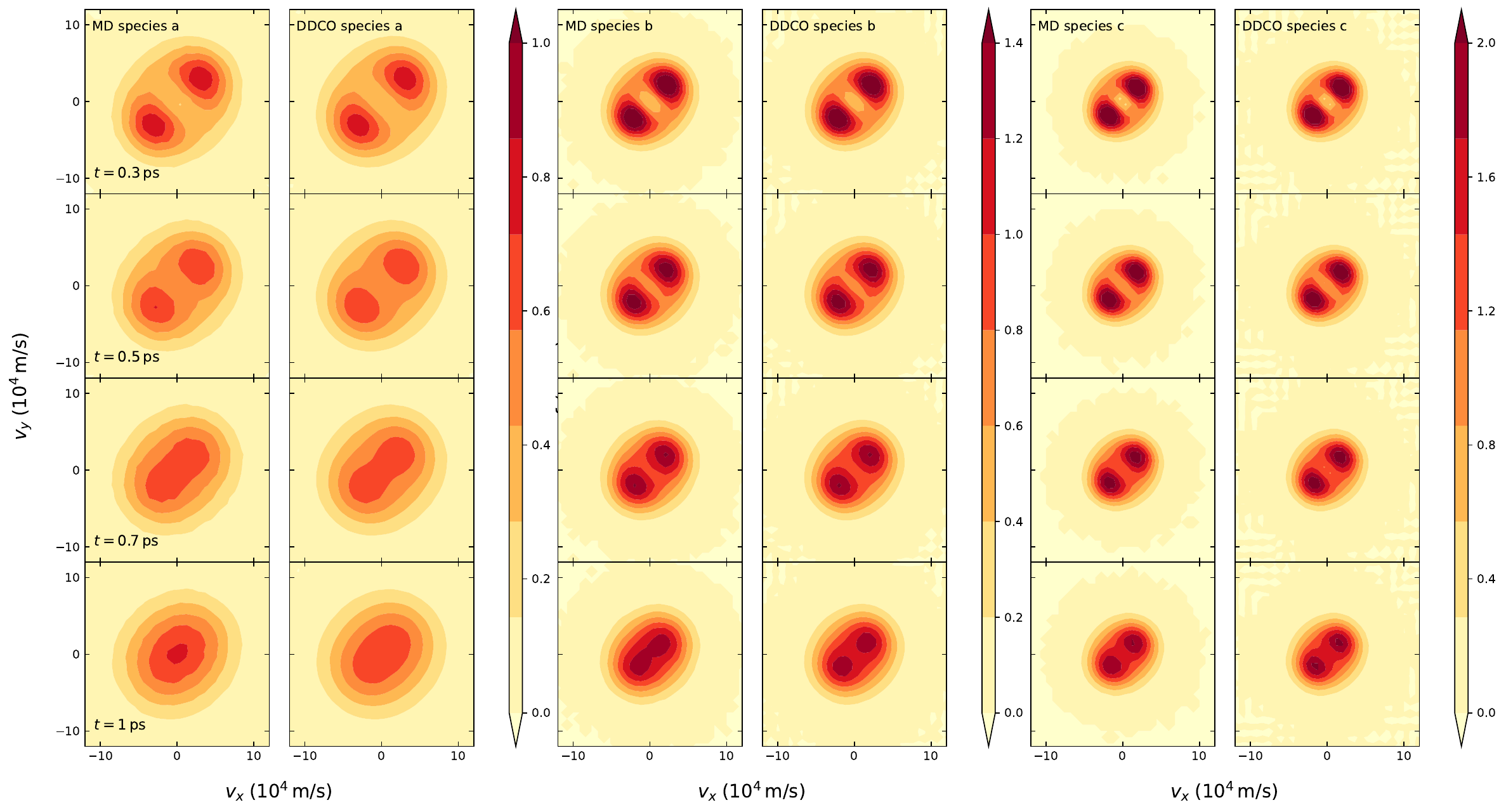}
    \includegraphics[width=0.48\linewidth]{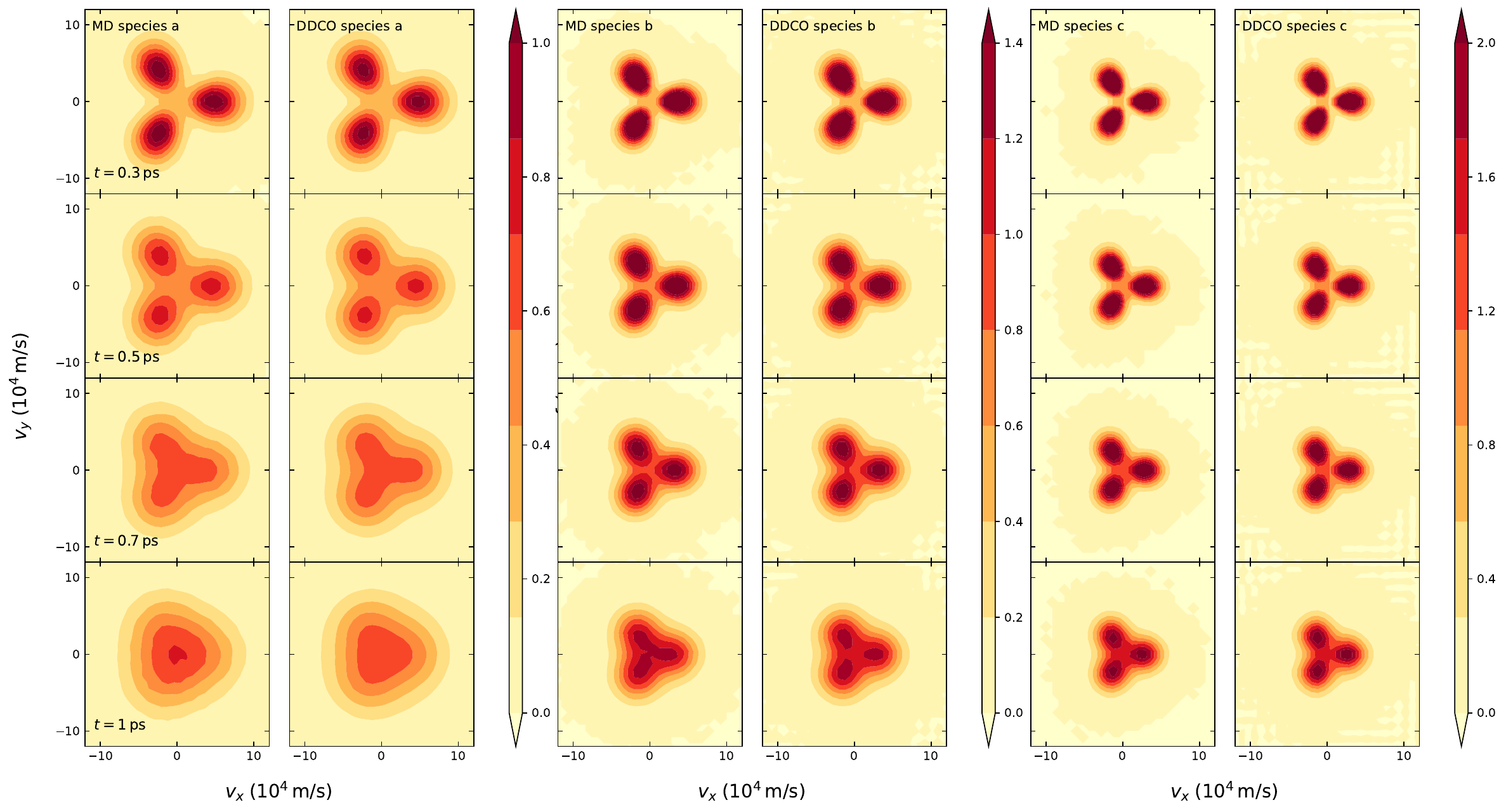}
    \caption{Comparison of the instantaneous velocity distribution on the $v_x$-$v_y$ plane for 4 initial PDFs at the same time in case 3.}
    \label{fig:mdf_3sp}
\end{figure}

\Cref{fig:mdf_1_3,fig:mdf_3sp} show the kinetic relaxation process for case 1 and case 3 defined in \Cref{sec:numerical_transport_coeff}. 
For both cases, DDCO reproduces the emergence of separated lobes, the persistence of asymmetric features at intermediate stage, and the isotropic unimodal final state.
The agreement for species with different masses and concentrations shows that the off-diagonal kernels transfer information consistently between PDFs, rather than fitting each species independently.
The smoother DDCO contours in low-probability regions are expected because MD contours retain finite-sample noise whereas the kinetic solution evolves a grid-based density.

\subsection{Comparison of instantaneous velocity distribution with the symmetric kernel and data-driven Landau form}

We then examine the role of velocity-permutation asymmetry in the off-diagonal collision kernels.
For case 2 defined in \Cref{sec:numerical_transport_coeff}, we compare the relaxation processes of MD with the full DDCO, a permutation-symmetric DDCO, and a modified Landau model.  Here, the plasma coupling parameter $\Gamma \sim \mathcal{O}(1)$ (i.e., $\ln \Lambda < 0$)  and the standard Landau model can not be directly applied. 
Accordingly, we fit the prefactor of the Landau kernel by relaxation times reported in \Cref{subsec:1sp_relax_time,sec:numerical_transport_coeff}, where the obtained model is named the data-driven Landau model. 
We impose symmetry conditions $\bm\omega_{st}(\bm c, \bm c') = \bm\omega_{st}(\bm c', \bm c)$, $s\neq t$ in the symmetric-DDCO during both training and simulation.

\Cref{fig:mdf_4methods} compares the relaxation of the initial symmetric double-well distribution.
The DDCO accurately predicts the locations and relative separation of the four non-Gaussian lobes and relaxation processes observed in MD for both species.
The permutation-symmetric DDCO retains a four-fold structure particularly for species $b$, whereas the Landau model yields different merging rates during the initial collision stage.
Although all three kinetic models approach the common equilibrium, they exhibit significant differences in transient dynamics. 
The present DDCO model without imposing the exchange symmetry constraints in Eq. \cref{eq:permutation_non_symmetry} 
is crucial for capturing the distinct responses of the different species due to unresolved correlations in the moderately coupled regime.

\begin{figure}[htbp]
    \centering
    \includegraphics[width=0.85\linewidth]{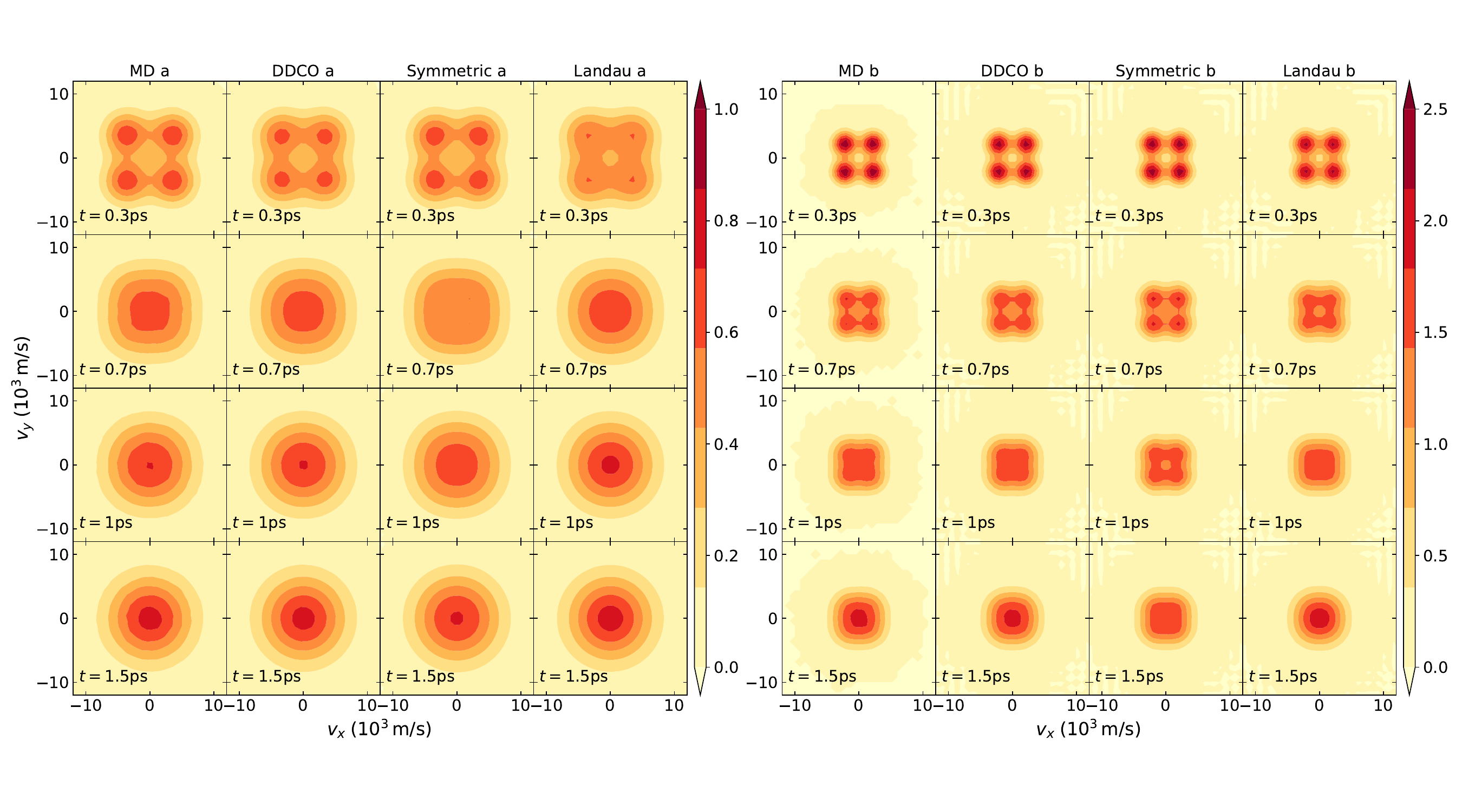}
    \caption{Comparison of the instantaneous velocity distribution on the $v_x$-$v_y$ plane predicted by MD, DDCO, permutation-symmetric DDCO, and the modified Landau model for case 2 with the symmetric double-well initial state, at $t=0.3$, $0.7$, $1.0$, and $1.5\,\mathrm{ps}$.}
    \label{fig:mdf_4methods}
\end{figure}

\subsection{Discrete structure, convergence, and computational cost}

For case 1 with the symmetric double-well initial state, we examine the relative error 
$\varepsilon_{X} = |X(t)-X(0)|/|X(0)|$ for the species mass and the total energy, and $\varepsilon_{\bm{P}}(t) = \left|\bm{P}(t)-\bm{P}(0)\right| / \sqrt{2\left|\mathcal{M}_{tot}(0)\right|\left|\mathcal{E}(0)\right|}$ for the total momentum.
\Cref{fig:MPE_etp_1_3} shows species mass and total momentum errors at approximately $10^{-16}$ and total energy error approximately $10^{-15}$, consistent with the conservation in \Cref{prop:discrete-structure}.
The entropy increases monotonically for every tested initial PDF and approaches a plateau.
This observation is consistent with the semi-discrete entropy identity, while forward Euler does not guarantee fully discrete entropy monotonicity without a positivity and time-step condition.

\begin{figure}[htbp]
    \centering
    \includegraphics[width=0.4\linewidth]{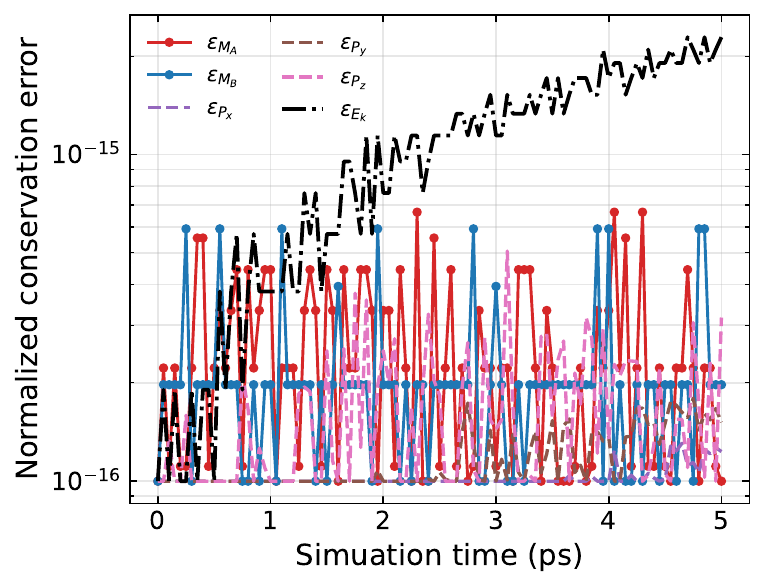} 
    \includegraphics[width=0.4\linewidth]{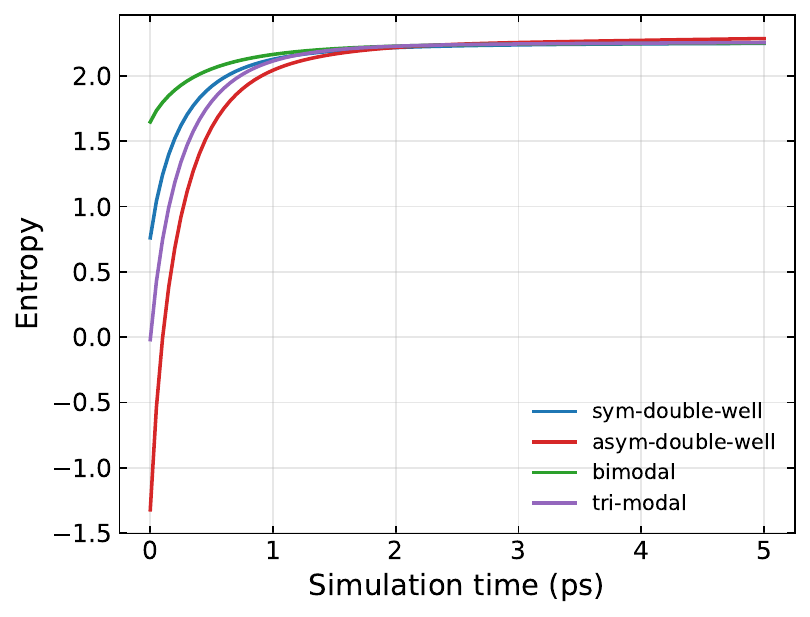}
    \caption{Left: normalized species mass, total momentum, and kinetic-energy errors for Case 1 with the symmetric double-well initial PDF.
    Right: discrete entropy for the four initial PDFs in \Cref{fig:mdf_1_3}.}
    \label{fig:MPE_etp_1_3}
\end{figure}

\Cref{fig:rel_err} shows the relative $L_2$ error for the numerical solution, where the solution with $N_{v}=400$ in each direction is used as the reference solution. The errors follow the expected second-order convergence rate.

\begin{figure}[htbp]
    \centering
    \includegraphics[width=0.4\linewidth]{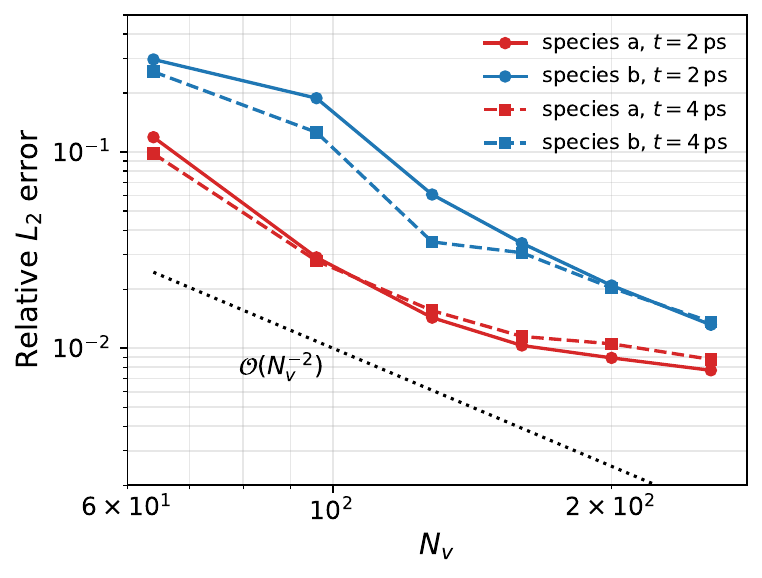}
    \includegraphics[width=0.4\linewidth]{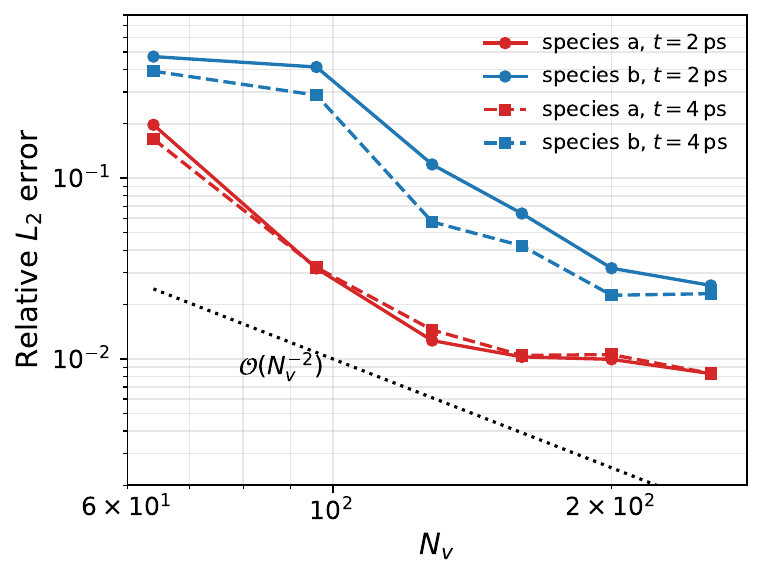}
    \caption{Relative $L_{2}$ errors at $t=2$ and $4\,\mathrm{ps}$ versus the number $N_{v}$ of nodes per velocity direction for the symmetric (left) and asymmetric (right) double-wells initial conditions.}
    \label{fig:rel_err}
\end{figure}

Finally, we examine the computational efficiency of the present DDCO model \cref{eq:multi-kinetic}. The direct evaluation of \cref{eq:multi-kinetic} with non-stationary $\bm\omega (\bm v, \bm v')$ requires $\mathcal O(N^{2})=\mathcal O(N_v^{6})$ computational cost.
As shown in \Cref{fig:time}, the separation structure formulation \cref{eq:two-encoder-kernel,eq:separation} enables the efficient FFT evaluation with $\mathcal O(N\log N) = \mathcal{O}(N_v^{3}\log N_{v})$ computational cost. 

\begin{figure}[htbp]
    \centering
    \includegraphics[width=0.4\linewidth]{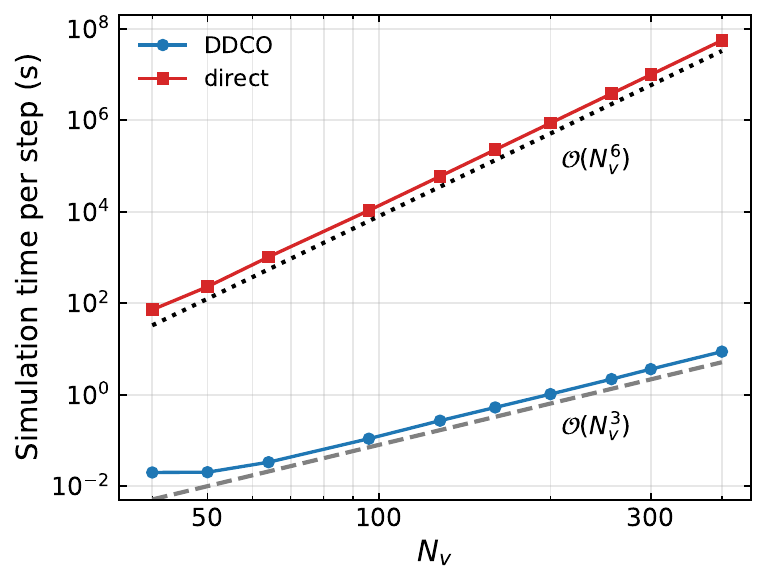}
    \caption{Wall-clock time per step for direct six-dimensional quadrature and the separated FFT implementation, with reference slopes of $\mathcal O(N_v^{6})$ and $\mathcal O(N_v^{3})$.}
    \label{fig:time}
\end{figure}

\section{Conclusions}\label{sec:summary}

In this work, we developed a structure-preserving, data-driven generalized collision operator for spatially homogeneous multi-species kinetic systems.
Within a local and pointwise-identifiable kernel class, we established necessary and sufficient kernel conditions for species-mass, total-momentum, and total-energy conservation, nonnegative entropy production, common-Maxwellian stationarity, and $O(3)$ frame indifference.
A central result is that ordered cross-species blocks must satisfy reciprocity but permit a non-symmetric structure under the permutation of the two velocity variables.
This structural freedom in Eq. \cref{eq:permutation_non_symmetry} differs from the classical Landau operator and direct extensions of one-species data-driven collision kernels \cite{zhao2025data, zhao2026fast}. 
In fact, this structure seems unique for the present model and is absent in other classical multi-species models such as the Balescu-Lenard \cite{lenard1960bogoliubov, balescu1960irreversible} and the Boltzmann \cite{boltzmann1872weitere, wild1951boltzmann} collision operators.

The admissible kernels were represented using scalar encoder functions that preserve the required physical structure by construction.
A low-rank kernel representation and random pair sampling enable efficient evaluation with complexity $\mathcal{O}(S^{2}J^{2}N\log N)$ and kernel training directly from MD trajectories.
The centered discretization preserves species mass, total momentum, and total kinetic energy under the fully discrete forward-Euler update and satisfies a nonnegative semi-discrete entropy-production identity.

The numerical results validate the model in both near-equilibrium and strongly non-equilibrium states.
The DDCO recovers the weak-coupling Landau behavior, captures temperature relaxation, and accurately predicts the diffusion and shear viscosity of multi-species systems in the moderately coupled regime, where the Landau model shows limitations. 
The learned operator also reproduces anisotropic relaxation and the evolution of non-Gaussian velocity distributions in two- and three-species mixtures.
In particular, the comparison with the permutation-symmetric DDCO and the data-driven Landau model demonstrates that the generalized ordered cross-species kernel is essential for resolving transient inter-species dynamics beyond weak coupling.
Future work will extend the framework to spatially inhomogeneous multi-species systems and incorporate local thermodynamic state dependence and self-consistent fields.

\appendix

\section{Molecular-dynamics and weak-form learning process}
\label{app:md-learning}

\paragraph{MD trajectories}
The multi-species system contains $(N_{a},N_{b}) = (10^{6}, 3 \times 10^{6})$, $(N_{a},N_{b}) = (3 \times 10^{6}, 10^{6})$, $(N_{a},N_{b},N_{c}) = (10^{6},10^{6},10^{6})$ particles for case 1, 2, 3 respectively, in a periodic cubic box of side length $10^{4}\,\text{\AA}$, and the species masses and charges are given in \Cref{sec:numerical_transport_coeff}.
Particles interact via pure Coulomb interactions, and the non-equilibrium trajectories are propagated in the NVE ensemble with initial equilibrium configurations and various velocity distributions.

\paragraph{Initial velocity distributions}
The initial velocity distributions for each species include uniform, Gaussian, bi-Maxwellian, symmetric double-well, asymmetric double-well, diagonal bimodal, and trimodal distributions, such as
\begin{equation}\label{eq:implemented-initial-distributions}
\begin{aligned}
	f_{s} &\sim \mathcal{U}[\alpha_{1},\alpha_{2}], &&\text{uniform},\\
	f_{s} &\sim \mathcal{N}_{3}([\alpha_{1},\alpha_{2},\alpha_{3}],\operatorname{diag}(\sigma_{1}^{2},\sigma_{2}^{2},\sigma_{3}^{2})), &&\text{Gaussian/bi-Maxwellian},\\
	f_{s} &\sim \prod_{i=1}^{3} \left[ \mathcal{N}(\alpha_{i},\sigma_{i}^{2}) + \mathcal{N}(\beta_{i},\gamma_{i}^{2}) \right], &&\text{double-well},\\
    f_{s} &\sim \mathcal{N}_{3}(\bm{\alpha},\sigma^{2} \bm{I}) + \mathcal{N}_{3}(-\bm{\alpha},\sigma^{2} \bm{I}), &&\text{diagonal bimodal},\\
    f_{s} &\sim \sum_{\ell=1}^{3} \alpha_{\ell} \mathcal{N}_{3}(\bm{\beta}_{\ell}, \sigma_{\ell}^{2} \bm{I}), &&\text{trimodal},
\end{aligned}
\end{equation}
with some parameters $\alpha$, $\beta$, $\gamma$ to ensure zero total momentum and the preset total kinetic energy.
The uniform and bi-Maxwellian cases are used in the training set, while trajectories with initial double-well, diagonal bimodal and trimodal distributions are used for validation.

\paragraph{Test functions}
The test functions $\psi_{k}$ in the loss function \cref{eq:loss} include radial functions, localized Gaussian functions, axis-odd functions and traceless quadratic functions, such as
\begin{equation}\label{eq:test-functions-training}
\begin{aligned}
    &\exp(-\alpha c^{2}), \quad c^{2}\exp(-\alpha c^{2}), \quad \exp[-(c^{2}-\beta_{0})^{2}], \quad \exp(-\alpha |\bm{c}-\bm{\gamma}|^{2})\\
    &c_{x} \exp(-\alpha c^{2}), \quad c_{y} \exp(-\alpha c^{2}), \quad c_{z} \exp(-\alpha c^{2}), \\
    &(c_{x}^{2}-c_{y}^{2})\exp(-\alpha c^{2}), \quad (c_{x}^{2}-c_{z}^{2})\exp(-\alpha c^{2}), \quad c_{x} c_{y} \exp(-\alpha c^{2})
\end{aligned}
\end{equation}
with $c=|\bm{c}|$ and some parameters $\alpha$, $\beta_{0}$, $\bm{\gamma}$ to cover the relevant regions of velocity space distributions.
The test functions are used to learn isotropic and anisotropic relaxation, asymmetric structures, momentum transfer, diffusion and shear modes.

\section*{Data and code availability}
The source code, input files, and trained models of this paper are publicly available at
\url{https://github.com/yuezhao1997/DDCO_multi_species}.

\section*{Acknowledgments}
This work was supported in part by the National Science Foundation under grant DMS-2143739, the Department of Energy under grant DOE-DESC0023164, and the ACCESS program through allocation MTH210005.
The authors also acknowledge the support provided by the Institute for Cyber-Enabled Research at Michigan State University.


\end{document}